\documentclass[11pt]{article}
\usepackage{authblk}
\usepackage{amssymb}
\usepackage{amsmath}
\usepackage{latexsym}
\usepackage{leqno}
\usepackage{theorem}
\newtheorem{theorem}{Theorem}[section]

\newtheorem{lemma}[theorem]{Lemma}
\newtheorem{corollary}[theorem]{Corollary}
\newtheorem{defi}[theorem]{Definition}

\newtheorem{rema}[theorem]{Remark}
\newtheorem{exam}[theorem]{Example}
\newenvironment{definition}{\begin{defi}\rm}{\hfill $\lhd$\end{defi}}
\newenvironment{remark}{\begin{rema}\rm}{\end{rema}}
\newenvironment{example}{\begin{exam}\rm}{\end{exam}}

\newenvironment{proof}{\begin{trivlist}\item[]{\bf
Proof.}}{\hfill {\sc qed}\end{trivlist}}

\newtheorem{claim2}{\sc Claim}

\renewcommand{\phi}{\varphi} % nicer \phi
\usepackage[all]{xy}
\usepackage{setspace}
\usepackage{comment}
\usepackage{tikz}
\usepackage{times}
\usepackage{helvet}
\usepackage{courier}
\usepackage{comment}
\usepackage{times}
\usepackage{helvet}
\usepackage{courier}
\usepackage{amssymb}
\usepackage{amsmath}
\usepackage{latexsym}
\usepackage{leqno}
\usepackage{theorem}
\usepackage{subfigure}
\usepackage{color}
\usepackage{tikz}
\usepackage{calc}
\usepackage{ifthen}
\usepackage[all]{xy}
\usetikzlibrary{arrows,positioning}
\usepackage[T1]{fontenc}
\usepackage[utf8]{inputenc}
\usepackage{algorithm}
\usepackage{algorithmic}

\usepackage[english]{babel}
\tikzset{
    vertex/.style = {
        circle,
        fill            = black,
        outer sep = 2pt,
        inner sep = 1pt,
    }
}
\newcommand{\RomanNumeralCaps}[1]
    {\MakeUppercase{\romannumeral #1}}
\date{}
\title{Local Differential Privacy for Belief Functions}

\author[1]{Qiyu Li}
\author[1]{Chunlai Zhou\thanks{Corresponding author: czhou@ruc.edu.cn}}
\author[1]{Biao Qin}
\author[2]{Zhiqiang Xu}
\affil[1]{Computer Science Dept., Renmin University of China, Beijing, CHINA}
\affil[2]{Mohamed bin Zayed University of Artificial Intelligence, Abu Dhabi, UAE}
\affil[ ]{\{qiyuli,czhou,qinbiao\}@ruc.edu.cn, zhiqiangxu2001@gmail.com}

\begin{document}
\maketitle

\begin{abstract}
   In this paper, we propose two \emph{new} definitions of local differential privacy for belief functions. One is based on  Shafer's semantics of randomly coded messages and the other from the perspective of imprecise probabilities. We show that such basic properties as composition and post-processing also hold for our new definitions. Moreover, we provide a hypothesis testing framework for these definitions and study the effect of ``don't know" in the trade-off between privacy and utility in discrete distribution estimation.
\end{abstract}

%&&&&&&&&&&&&&&&&&&&&&&&&&&&&&&&&&&&&&&&&&&&&

\section{Introduction}

 \emph{Differential privacy} (DP) is a mathematically rigorous definition of privacy which addresses the paradox of learning nothing about
an \emph{individual} while learning useful information about a \emph{population} \cite{DworkMNS06,DworkR14}.  In particular, \emph{local differential privacy} (LDP) is a model of differential privacy with the added restriction that even if an adversary has access to the personal responses of an individual in the database, that adversary will still be unable to learn too much about the user's personal data \cite{KLNRS08,KairouzOV16,DuchiJW13}.  The uncertainty in standard LDP mechanisms is usually provided by randomization which associates each input with a \emph{probability} function over all possible outputs. The prototypical example of an LDP mechanism is the \emph{randomized response} survey technique proposed in \cite{Warner65}.  Current randomized response mechanisms equate privacy-preserving with lying and are designed on the assumption that users abide by the data collection protocol which allows respondents to lie with a \emph{known} probability.  However, recent research results from the perspective of the \emph{respondents} show that, in practice, although these mechanisms allow the respondents to maintain privacy, the procedures may confuse respondents, fail to address the concerns of the users  and hence yield
nonresponse or noncompliance \cite{Xiong0LJ20,Cummings2021need,Ramokapane2021truth}.  An effective  differential privacy communication can increase data-sharing rates \cite{Xiong0LJ20}.  
%Despite the methodological advances in randomized response \cite{GraemeIZ15} to increase cooperation, some respondents may still feel uncomfortable under the design of randomized response tools and they do not really want to answer or just say ``I don't know" \cite{Cummings2021need}.

To address noncompliance and nonresponse,  we propose in this paper to design  differential privacy mechanisms which incorporate ``don't know" or nonresponse as an alternative outcome or allow imprecision in the mechanism design.  In practice,  people may prefer not to response or say ``I don't know" to withhold sensitive information which minimizes the questionable ethical consequences of lying in their eyes \cite{BullekGMP2017}. By addressing such ethical privacy concerns, our new mechanisms aims to increase respondents' willing to share their data.  Here we study this new type of privacy mechnisms from a more general Dempster-Shafer perspective by representing uncertainty in privacy mechanisms with \emph{belief functions} \cite{Dempster67,Shafer76}.  The Dempster-Shafer theory (also known as the theory of evidence or the theory of belief functions) is a well-known uncertainty theory for its expressiveness in representing ignorance. The theory improves the root concepts of  probabilities ``yes" and ``no" that sum to one, by appending a third probability of ``\emph{don’t know}" \cite{Dempster08}.  As the world of statistical analysis moves more and more to ``big data" and associated ``complex systems", the Dempster-Shafer theory provides a middle ground with the third probability ``don't know" and can be expected to become increasingly important in privacy protection.

Our first and main contribution in this paper is to propose two new definitions of LDP (one is $\epsilon$-local differential privacy according to Shafer ($\epsilon$-SLDP) (Definition \ref{def:SLDP}) and the other according to Walley ($\epsilon$-WLDP) (Definition \ref{def:WLDP})) and to provide a statistical framework for these two definitions as the trade-offs between type \RomanNumeralCaps{1} and \RomanNumeralCaps{2} errors in a natural hypothesis-testing problem (Theorems \ref{thm:equivalence4SLDP} and \ref{thm:equivalence4WLDP}). Our second contribution is to characterize the effect of ``don't know" in the trade-off between privacy and utility in discrete distribution estimation problem. The privacy  mechanisms in the two definitions associate each input $x$ with a \emph{belief function} on the output set $Y$. The difference between these two definitions comes from their different semantics of belief functions.   The first definition is motivated by Shafer's interpretation of belief functions as randomly coded messages \cite{ShaferT85}.  In this semantics, we generalize Warner's randomized response mechanism by allowing answering ``don't know" with probability $1-p-q$ where $p$ is the probability of answering truthfully and $q$ the probability of lying.  For the discrete distribution estimation problem of a generalized Warner's model, we study the effect of ``don't know" on the trade-off between the privacy loss and the estimation accuracy. The \emph{most important and difficult} step is to compute the variance of the maximum likelihood estimation of the parameter $\pi$, the true proportion of the people with the sensitive property.  We employ some combinatorial techniques to obtain a formula for the estimation accuracy (Theorem \ref{th:variance}).
 We show that, when the probability of ``don't know" increases, the overall effect of the trade-off for this generalized model decreases, and when this probability equals 0, the effect is optimal and the trade-off is the same as that for the standard Warner's model (Figure \ref{fig:2}).
    In the second definition, we adopt the imprecise-probability semantics to accommodate \emph{unknown response probabilities} in privacy mechanisms and interpret belief function $bel$ as the set of all  probability functions $pr$ which are consistent with $bel$ \cite{Walley90}. Both the privacy loss and estimation accuracy are defined with respect to those consistent probability functions according to the worst-case analysis.
 Moreover, we compare the trade-offs between privacy and estimation accuracy for these two definitions ($\epsilon$-SLDP and $\epsilon$-WLDP) and Warner's randomized response mechanism (Figure \ref{fig:tradeoffs}).

\section{Dempster-Shafer Theory} \label{sec:DS}

Let $\Omega$ be a frame  and $\mathcal{A} = 2^{\Omega}$ be the Boolean algebra of propositions. $|A|$ denotes the cardinality of a subset $A$.  A \emph{mass assignment} (or \emph{mass function}) over $\Omega$  is a mapping $m: \mathcal{A}\rightarrow [0,1]$ satisfying    $ \sum_{A\in \mathcal{A}} m(A) =1$.
A mass function $m$ is called \emph{normal} if $m(\emptyset) =0$.   A \emph{belief function} is a function $bel: \mathcal{A}\rightarrow [0,1]$ satisfying the conditions: $bel(\emptyset) =0$,  $bel(\Omega) =1$ and
  $bel(\bigcup_{i=1}^{n} A_i)\geq\sum_{\emptyset \neq I \subseteq \{1, \cdots, n\}} (-1)^{|I|+1} bel(\cap_{i\in I} A_i)$
   where $A_i \in \mathcal{A}$ for all $i\in \{1, \cdots, n\}$.
A mapping $f: \mathcal{A} \rightarrow [0,1]$ is a belief function if and only if its M{\"o}bius transform is a mass assignment (Page 39 in \cite{Shafer76}). In other words, if $m: \mathcal{A}\rightarrow [0,1]$ is a mass assignment, then it determines a belief function $bel: \mathcal{A}\rightarrow [0,1]$ as follows:
    $bel(A) = \sum_{B \subseteq A} m(B)$ for all $A\in \mathcal{A}$.
Moreover,  given a belief function $bel$, we can obtain its corresponding mass function $m$ as follows:
   $ m(A)  = \sum_{B\subseteq A} (-1)^{|A\setminus B|} bel(B) \text{ for all } A\in \mathcal{A}.$
Intuitively, for a subset event $A$,  $m(A)$ measures the belief that an agent commits \emph{exactly} to $A$, not the total belief $bel(A)$ that an agent commits to $A$.  A subset $A$ with non-zero mass is called a \emph{focal set}.
The belief function $bel$ is called \emph{Bayesian} if $m(A)=0$ for all non-singletons $A$.  The corresponding \emph{plausibility function} $pl_m: 2^{\Omega} \rightarrow [0,1]$ is defined by $pl_m(A) = \sum_{E\cap A \neq \emptyset} m(E) $ for all $A\subseteq \Omega$.  Whenever the context is clear, we drop the subscript $m$. For $m, bel$ and $pl$, if we know any one of them, then we can determine the other two.   Without further notice, all mass functions in this paper are assumed to be normal and all subsets are focal.

% If $m_1$ and $m_2$ are two mass functions on $\Omega$ induced by two \emph{independent} evidential sources, the combined mass function is calculated according to \emph{Dempster's rule of combination}: for any $C\subseteq \Omega$,
%\begin{small}
%\begin{align}
%    (m_1 \oplus m_2) (C) & = \frac{\sum_{A \cap B = C} m_1(A) m_2(B)}{\sum_{A %\cap B \neq \emptyset} m_1(A) m_2(B)} \label{Combination}
%\end{align}
%\end{small}

  In this paper, we focus on only two semantics of belief functions. The first one is
   Shafer's semantics of belief functions in terms of \emph{randomly coded messages}. Suppose someone chooses a code at random from a list of codes, uses the code to encode a message, and then sends us the result.
We know the list of codes and the chance of each code being chosen--say
the list is $c_1, \cdots, c_n$, and the chance of $c_i$ being chosen is $p_i$. We decode the
encoded message using each of the codes and find that this always produces a
message of the form ``the truth is in A" for some non-empty subset $A$ of the
set of possibilities $\Omega$. Let $A_i$ denote the subset we get when we decode using
$c_i$, and set
$m(A) = \sum\{p_i: 1\leq i\leq n, A_i=A\}$
for each $A\subseteq \Omega$.   The number $m(A)$ is the sum of the chances for those codes that indicate A was the true message; it is, in a sense, the total chance that
the true message was $A$. Notice that $m(\emptyset) = 0$ and that the $m(A)$ sum to one. The quantity
$ bel(A) = \sum_{B\subseteq A} m(B)$
is, in a sense, the total chance that the true message implies $A$. If the true message
is infallible and the coded message is our only evidence, then it is natural
to call $bel (A)$ our probability or degree of belief that the truth lies in $A$.
The second interpretation of belief functions in this paper is from the perspective of imprecise probabilities. Given a belief function $bel$, let  $\mathcal{P}_{bel}$ denote the set of all probability functions which are consistent with or dominate over $bel$. In other words,
$ \mathcal{P}_{bel} = \{pr: pr\text{ is a probability function on }\Omega$ and $pr \geq bel\}$ where $pr\geq bel $ means $pr(E) \geq bel(E)$  for all  $E \subseteq \Omega$.  Due to lack of information, uncertainty can't be represented by a probability function but by a belief function $bel$. All consistent probability functions are possible. Whenever enough information is available, we may specify a probability function from $\mathcal{P}_{bel}$ to represent the uncertainty. One may refer to \cite{Cuzzolin21} and \cite{DworkR14} for a detailed introduction to belief functions and  DP.

%&&&&&&&&&&&&&&&&&&&&&&&&&&&&&&&&&&&&&&&&&&&&&&&&&&&&&&&&&&&

\section{Local Differential Privacy}\label{sec:LDP}

Let $X$ be a private source of information defined on a discrete,
finite input alphabet $X = \{x_1, \cdots,x_k\}$ and $Y$ be an output alphabet
$Y = \{y_1, \cdots,y_l\}$ that need not be identical to the input alphabet
$X$. In this paper, we will represent a privacy
mechanism $Q$ via a row-stochastic matrix. For simplicity, we also use $Q$ to denote this matrix.  $Q$ is called an \emph{evidential} privacy mechanism if each row of the matrix $Q$ is a mass function on $Y$.  In other words, each evidential
privacy mechanism $Q$
maps $X = x$ to $Y\in E$ with $Q(x)$ which can be represented by a mass $m^{Q}_x(E)$ (belief $bel^{Q}_x(E)$ or plausibility $pl^{Q}_x(E)$) where $m^{Q}_x$ ($bel^{Q}_x(E)$ or $pl^{Q}_x(E)$) is a mass (belief or plausibility) function on $Y$ for all $x\in X$. Since $m_x^Q(\emptyset) =0$ for all $x$, we write the mechanism $Q$ as a $k\times (2^l-1)$ matrix. Whenever the context is clear, we usually drop the superscript $Q$.  In this paper, we assume that all the alphabet sets are finite. In other words, an evidential privacy mechanism is just a standard LDP mechanism  whose instructions are defined by random \emph{sets} instead of probability functions.

\subsection{LDP according to Shafer}\label{subsec:LDPS}

For an evidential privacy  mechanism $Q$, let $r^Q_S = max_{x,x'\in X, E\subseteq Y}\frac{m_x^Q(E)}{m_{x'}^Q(E)}$ and $\epsilon^Q_S = ln(r^Q_S)$.

\begin{definition}\label{def:SLDP} For any $\epsilon > 0$,
the mechanism $Q$ is called \emph{$\epsilon$-locally differential private} according to Shafer ($\epsilon$-SLDP for short) if $ -\epsilon \leq \epsilon^Q_S \leq \epsilon$.  
And  $\epsilon^Q_S$ is called the \emph{privacy loss} of $Q$ according to Shafer and $\epsilon$ is a \emph{privacy budget}.
\end{definition}
In other words, by observing $E$, the adversary cannot reliably infer whether $X = x$ or
$X = x'$ (for any pair $x$ and $x'$). Indeed, the smaller the $\epsilon$
is, the closer the likelihood ratio of $X = x$ to $X = x'$ is to
1. Therefore, when $\epsilon$ is small, the adversary cannot recover
the true value of $X$ reliably.  In this definition, we adopt Shafer's interpretation as randomly coded messages. Each subset of $Y$ is treated as an individual message or response. The mechanism randomly chooses a code $c$ and uses it to encode a message $E\subseteq Y$.   And $m_x(E)$ is equal to the chance of choosing $c$. If we set $2^Y\setminus \{\emptyset\}$ as the output alphabet, then the above $Q$ is simply the standard local differential private mechanism.
In particular, if each row of $Q$ is Bayesian, then $Q$ is essentially a standard randomized mechanism and the $\epsilon$-SLDP is just the standard $\epsilon$-$LDP$ for randomized privacy mechanisms.
Almost all basic properties for privacy-preserving  randomized mechanisms can be generalized to the setting of belief functions.   Let $r^Q_{pl,S} = max_{x,x'\in X, E\subseteq Y}\frac{pl_x^Q(E)}{pl_{x'}^Q(E)}$ and $r^Q_{bel,S} = max_{x,x'\in X, E\subseteq Y}\frac{bel_x^Q(E)}{bel_{x'}^Q(E)}$. Denote $\epsilon^Q_{pl,S}:= ln(r^Q_{pl,S})$ and $\epsilon^Q_{bel,S}:= ln(r^Q_{bel,S})$. 

\begin{lemma} \label{le:alternation1} If privacy mechanism $Q$ is $\epsilon$-$SLDP$, then  $-\epsilon \leq \epsilon^Q_{bel,S} \leq \epsilon$ and
      $-\epsilon \leq \epsilon^Q_{pl,S} \leq \epsilon$.
\end{lemma}
%\begin{proof} These two propositions follow directly from the inequality (\ref{eq:def_SLDP}) and the facts that %$bel_x(E) = \sum_{E'\subseteq E} m(E')$ and $pl_x(E) = \sum_{E'\cap E \neq \emptyset} m(E')$.%
%\end{proof}

 From Lemma \ref{le:alternation1}, we know that $\epsilon^Q_S \geq \epsilon^Q_{pl,S}$. But generally we don't have the converse that $\epsilon^Q_{pl,S} \geq \epsilon^Q_{S}$. If we have several building blocks for designing differentially
private algorithms, it is important to understand how we can combine
them to design more sophisticated algorithms. 

\begin{lemma} \label{le:composition}(Composition) Let $Q_1$ be an $\epsilon_1$-SLDP  evidential privacy mechanism from $X$ to $Y_1$ and $Q_2$ be an $\epsilon_2$-SLDP  evidential privacy mechanisms from $X$ to $Y_2$.  Then their combination $Q_{1,2}$ defined by $Q_{1,2}(x) = (Q_1(x),Q_2(x) )$ is $\epsilon_1 + \epsilon_2$-SLDP.
\end{lemma}
%\begin{proof} The composition follows from the fact that $\frac{m^{Q_{12}}_x(E_1, E_2)}{m^{Q_{12}}_{x'}(E_1, %E_2)}=\frac{m^{Q_{1}}_x(E_1)}{m^{Q_{1}}_{x'}(E_1)}\frac{m^{Q_{2}}_{x}(E_2)}{m^{Q_{2}}_{x'}(E_2)}$.
%\end{proof}

%\begin{lemma}
%(Simple composition) Let $\mathcal{A}_1$ be an $\epsilon_1$_LDP from $X_1$ to %$X_2$ and $\mathcal{A}_2$ be an $\epsilon_2$-LDP from $X_2$ to $X_3$.  Then the %simple composition $\mathcal{A}_2\circ \mathcal{A}_1$ from $X_1$ to $X_3$ is an %$\epsilon_1 + \epsilon_2$-LDP.
%\end{lemma}
%[]ADD Proof]

The composition of a \emph{data-independent} mapping $f$ with an $\epsilon$ locally
differential private algorithm $Q$ is also $\epsilon$ locally  differential private.

\begin{lemma} \label{le:post}(Post-processing) Let $Q$ be an $\epsilon$-SLDP mechanism from $X$ to $Y$ and $f$ is a randomized algorithm from $Y$ to another finite alphabet set $Z$. Then $f\circ Q$ is an $\epsilon$-SLDP mechanism from $X$ to $Z$.
\end{lemma}
%\begin{proof}  The proof for $\epsilon$-SLDP algorithm is similar to that of the post-processing property for %probabilistic privacy mechanisms (Proposition 2.1 in \cite{DworkR14}). We prove the proposition only for a %deterministic function $f: Y\rightarrow Z$. For any $x, x'\in X$, $m_x^{f\circ Q} (E) = m_x^Q (f^{-1}(E)) \leq %e^{\epsilon} m_{x'}^Q (f^{-1}(E))$.

%\end{proof}

%&&&&&&&&&&&&&&&&&&&&&&&&&&&&&&&&7

Now we offer a \emph{hypothesis testing} interpretation for the above $\epsilon$-$SLDP$.  From an attacker's perspective, the privacy requirement can be formalized as the following hypothesis testing problem for two datasets $x$ and $x'$: 
\begin{center} 
	$H_0$: the underlying dataset is $x$ vs. $H_1$: the underlying dataset is $x'$.
\end{center}
The output of the mechanism $Q$ serves as the basis for performing the hypothesis testing problem. The distinguishability of the two inputs $x$ and $x'$ can be translated into the trade-off between type \RomanNumeralCaps{1} and  type \RomanNumeralCaps{2} errors \cite{DongRS21}.  For belief functions, it is natural to consider \emph{minimax tests} \cite{HuberS73}. Formally, consider a rejection rule $\phi: Y \rightarrow [0,1]$.   Let $\mathcal{P}_x^Q$ and $\mathcal{P}_{x'}^Q$ denote the two sets of probability functions dominating $bel^Q_x$ and $bel^Q_{x'}$ respectively. In other words, $\mathcal{P}_x^Q = \{pr\in \Delta(Y): pr \geq bel_x^Q\}$
and $\mathcal{P}_{x'}^Q = \{pr\in \Delta(Y): pr \geq bel_{x'}^Q\}$. The \emph{lower power} of $\phi$ under $x'$ is defined as $\pi_{x'}:=\text{inf}_{pr \in \mathcal{P}^Q_{x'}} \mathbb{E}_{pr}(\phi)$.
 In the setting of $\epsilon$-$SLDP$, we assume that type \RomanNumeralCaps{1} error $\alpha_{\phi}$ is represented by $\text{sup}_{pr \in \mathcal{P}^Q_{x}} \mathbb{E}_{pr}(\phi)$ and type \RomanNumeralCaps{2} error by $\beta_{\phi}=1-\text{inf}_{pr \in \mathcal{P}^Q_{x'}} \mathbb{E}_{pr}(\phi)$.  A test $\phi$ is called a \emph{level}-$\alpha$ \emph{minimax test} if $\phi=argmin\{\beta_{\phi}: \alpha_{\phi}\leq \alpha\}$.  The following theorem is a generalization of the well-known result (Theorem 2.4 in \cite{WassermanZ2010}) for standard differential privacy.

\begin{theorem}\label{thm:equivalence4SLDP} For any evidential privacy mechanism $Q$, the following two statements are equivalent:
	\begin{enumerate}
		\item $Q$ is $\epsilon$-SLDP;
		\item If type  \RomanNumeralCaps{1} error $\alpha_{\phi}\in [l, L]$, then type \RomanNumeralCaps{2} error 
$\beta_{\phi}\in [u(L), U(l)]$ where $u(\alpha): = max\{e^{-\epsilon}(1-\alpha), 1-\alpha e^{\epsilon}\}$ and $U(\alpha): = min\{e^{\epsilon}(1-\alpha),1-\alpha e^{-\epsilon}\}$.
	\end{enumerate}
\end{theorem}

Now we consider the hypothesis testing problem for the composition and would like to distinguish between $Q(x)\times Q(x)$ and $Q(x')\times Q(x')$. The corresponding type \RomanNumeralCaps{1} and \RomanNumeralCaps{2} errors $\alpha^2_{\phi}$ and $\beta^2_{\phi}$ can be defined similarly.  For simplicity, we only show the two-fold composition and other multi-fold compositions can be obtained similarly. 

\begin{corollary}\label{cor:composition4SLDP} For the hypothesis testing problem for the composition, if type  \RomanNumeralCaps{1} error $\alpha^2_{\phi}\in [l, L]$, then type \RomanNumeralCaps{2} error 
$\beta^2_{\phi}\in [u^2(L), U^2(l)]$ where $u^2(\alpha): = max\{e^{-2\epsilon}(1-\alpha), -\alpha+\frac{2}{e^{\epsilon}+1}, 1-\alpha e^{2\epsilon}\}$ and $U^2(\alpha): = min\{e^{2\epsilon}(1-\alpha),1-\alpha e^{-2\epsilon}, -\alpha +\frac{3-e^{-2\epsilon}}{e^{\epsilon}+1}\}$.
\end{corollary}

Both Theorem \ref{thm:equivalence4SLDP} and Corollary \ref{cor:composition4SLDP} can be visualized in Figure \ref{fig:1}.

\begin{figure}[htbp]
\centering
\includegraphics[height=6.0cm,width=6cm]{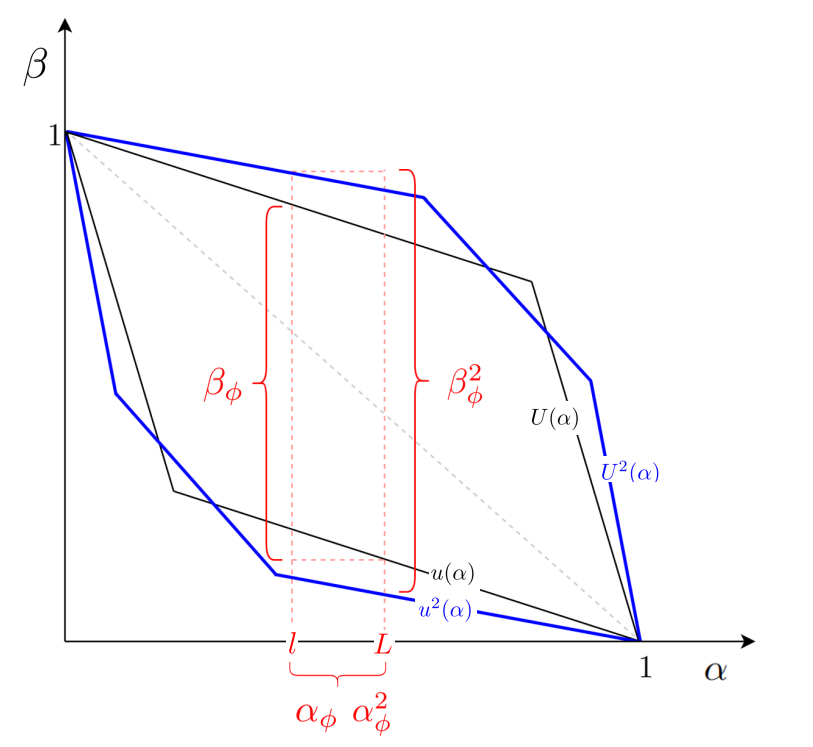}
\caption{Trade-off between type \RomanNumeralCaps{1} and \RomanNumeralCaps{2} errors for SLDP }
\label{fig:1}
\end{figure}

The discrete estimation problem is defined as
follows. Given a prior which is a vector $\pi = (\pi_1, \dots, \pi_k)$ on the probability
simplex $\mathbb{S}^k=\{p = (\pi_1, \dots, \pi_k): \pi_i \geq 0 (1\leq i\leq k), \sum_{i=1}^k \pi_i =1\}$, samples $X_1,\cdots, X_n$ are drawn i.i.d. according
to $\pi$. A privacy mechanism $Q$ is
then applied independently to each sample $X_i$ to produce
$Y^n = (Y_1; \cdots, Y_n)$, the sequence of private observations.
Observe that the $Y_i$’s are distributed according to $m = \pi Q$, which are mass functions not necessarily probability functions when $Q$ is evidential.  Our goal is to estimate the distribution vector $\pi$
from $Y^n$ within a certain privacy budget requirement. The performance of the estimation may be measured via a loss function. Here we use the mean square loss function.  $Q$ is called \emph{optimal} if the estimation error is the smallest.   A classic example for discrete distribution estimation is Warner's randomized response method for survey research \cite{Warner65}.

%Randomized response is a technique developed in the social sciences to collect statistical information
%about embarrassing or illegal behavior, captured by having a property $P$. For example, study participants are %told to report whether or not they have property $P$ according to a randomized mechanism specified as follows:
%\begin{enumerate}
%    \item Flip a (unbiased) coin;
%    \item If tails, then respond truthfully.
%    \item If heads, then flip a second coin and respond “Yes” if heads and
%“No” if tails.
%\end{enumerate}
%``Privacy” comes from the plausible deniability of any outcome; in particular, if having property $P$ %corresponds to engaging in illegal behavior,
%even a ``Yes” answer is not incriminating, since this answer occurs with
%probability at least $\frac{1}{4}$ whether or not the respondent actually has
%property $P$. Accuracy comes from an understanding of the noise generation procedure: The expected number of %“Yes” answers is
%$\frac{1}{4}$ times the number of participants who do not have property $P$ plus
%$\frac{3}{4}$ the number having property $P$. Let $X_W=\{P,\neg P\}$ and $Y_W=\{Yes, No\}$.  We can easily see %that the privacy loss equals $ln 3$ \cite{DworkR14}.

% The prototypical  Warner's randomized response mechanism $Q_W$ can be  illustrated in Table \ref{tab:table1}.
%\begin{table}[]
%    \centering
%    \begin{tabular}{c|c |c}
%         &  Yes & No\\
%         \hline
%        $P$ & $p$ & $1-p$\\
%        $\neg P$   & $1-p$ & $p$
%    \end{tabular}
%    \caption{Warner's Model}
%    \label{tab:table1}
%\end{table}
\begin{example}\label{ex:theme}
According to prototypical  Warner's randomized response mechanism $Q_W$, the respondent answers truthfully with probability $p$ and lies with probability $1-p$.  Let $\pi$ be the true proportion of the people having property $P$. A sample of $Y_1, \cdots, Y_n$ of respondents are drawn with replacement from the population and their responses are distributed  i.i.d. according to $(q_1, q_2)= (\pi, 1-\pi) Q_W$.    So $q_1= \pi p + (1-\pi)(1-p)$ and $q_2= \pi(1-p) + (1-\pi) p$. Arrange the indexing of the sample so that the first $n_1$ respondents say "Yes" and the remaining $n-n_1$ answers "No".  We obtain  the maximum likelihood estimation of $\pi$ as 
$\hat{\pi} = \frac{p-1}{2p-1} + \frac{n_1}{(p-1)n}$.
It can be shown \cite{Warner65,HolohanLM17} that this distribution estimation $\hat{\pi}$ is unbiased and its mean square error or variance is the following formula:
\begin{align}
	Var[\hat{\pi}] & = \frac{-(\pi-\frac{1}{2})^2+ \frac{1}{4}}{n} + \frac{\frac{1}{4(2p-1)^2}-\frac{1}{4}}{n}\label{Eq:variance}
\end{align}
Within the privacy budget of $\epsilon$, the optimal privacy mechanism  is
\begin{equation*}
Q_{WRR} = \frac{1}{e^{\epsilon}+1}\left(
\begin{array}{cc}
e^{\epsilon} & 1 \\
1 & e^{\epsilon}
\end{array} \right).
\end{equation*}

\end{example}

Now we are generalizing the above Warner's model by allowing a third response ``I don't know" and representing the corresponding uncertainty with a mass function.
  Let  $Q_{2\times 3}$ denote a known row-stochastic matrix as follows:

\begin{equation*}
Q_{2\times 3} = \left(
\begin{array}{ccc}
p & q & 1-p-q \\
q & p & 1-p-q
\end{array} \right)
\end{equation*}
where $p, q\in [0,1]$. $Q_{2\times 3}$ may be regarded as a generalized Warner's randomized response mechanism where a respondent answers truthfully with probability $p$, tells a lie with $q$ and don't respond or respond "I don't know" with probability $1-p-q$.  We may assume in this paper that $p>\frac{1}{2}$.

\begin{remark} In the following we choose to work with such a simple form $Q_{2\times 3}$ of LDP for belief functions. A more general form  can be studied similarly, but unfortunately we couldn't obtain closed forms for (approximate) estimation and error as we achieve below for this simple form $Q_{2\times 3}$. The maximum likelihood estimation problem for the more general form can be naturally formalized as a mixture of the conditional mass functions associated with the evidential privacy mechanism with the mixture proportions as the unknown  prior distribution of the sensitive population.We can apply EM algorithm to approximate the prior distribution and compute its Fisher information and further the standard error of the approximation \cite{Agrawal2001design}. However, the simple form  provides us with a neat formula of estimation error (Theorem \ref{th:variance}) and hence a formula for the privacy-utility trade-off.  Indeed the simple form for evidential mechanism  is enough to illustrate the effect of the answer ``I don’t know" or nonresponse on the privacy-utility trade-off. Both the simulation experiments and Figure 2 afterwards are based on the above analysis.
 In this paper we mainly focus on this simple form $Q_{2\times 3}$. But we expect that such a simple form to evidential privacy mechanisms is the same as Warner’s $2\times 2$ mechanism to the standard LDP. For standard LDP, every approximate DP algorithm can be simulated by a (leaky) variant of Warner’s  $2\times 2$ mechanism (a well-known result in optimal composition \cite{MurtaghV18,Kairouz2017composition}). From a broader and deeper perspective, we believe that every approximate evidential privacy mechanism can be simulated by some variant of our $2\times 3$ mechanisms in this paper.  In this sense, our contribution is similar to Warner’s contribution to standard LDP.
\end{remark}

A simple random sample of $n$ people is drawn with replacement from
the population. Let $Z_i$ denote the $i$-th sample element. Recall that $\pi$ is the true proportion of the people with the sensitive property $P$. $Z_i$ is distributed according to the following $(q_1, q_2, q_3)$:

\[
\left(
\begin{array}{ccc}
q_1 & q_2 & q_3
\end{array}
\right)
=\left(
  \begin{array}{ccc}
    \pi & 1-\pi
      \end{array}
\right)
\left(
\begin{array}{ccc}
p & q & 1-p-q \\
q & p & 1-p-q
\end{array}
\right)
\]
In other words, $q_1 =\pi p +(1-\pi)q$, $q_2 = \pi q + (1-\pi) p$, and $q_3 = 1-p -q$.
 Note that $q_1+q_2+q_3=1$. It implies that $Z_i$ says ``Yes", ``No" and ``don't know" with probabilities $q_1, q_2$ and $q_3$ respectively. Arrange the indexing of the sample so that  the first $n_1$ sample elements say $Yes$, the next $n_2$ say $No$ and the last $n_3$ say ``don't know" where $n_1, n_2$ and $n_3$ are natural numbers such that $n_1 + n_2 +n_3 =n$.  So the likelihood of the sample is
	$L(\pi) = q_1^{n_1} q_2^{n_2} q_3^{n_3}$.  By taking its logarithm and then setting its derivative to be zero, we obtain $\frac{n_1}{q_1} - \frac{n_2}{q_2} =0$.  So we obtain the maximum likelihood estimation (MLE) of $\pi$ as follows:
\begin{align}
	\hat{\pi} & = \frac{n_2 q- n_1p}{(n_1+n_2)(q-p)}.
\end{align}
Now we want to compute the expectation of $\hat{\pi}$. From $Z_i$, we define three new random variables $Z_{i1} = \mathbb{I}_{[Z_i = Yes]}, Z_{i2} = \mathbb{I}_{[Z_i = No]}$ and $Z_{i3} = \mathbb{I}_{[Z_i = \text{don't know}]}$ (where $\mathbb{I}$ denotes the indicator function). Then $Z_i= Z_{i1}+ Z_{i2}+Z_{i3}, N_1 =  \sum_{i=1}^n Z_{i1}, N_2 = \sum_{i=1}^n Z_{i2}$ and $N_3 =  \sum_{i=1}^n Z_{i3}$. So $N_1+ N_2+N_3 =n$. We obtain the conditional expectation of the MLE.

\begin{comment}
\begin{lemma}
	$\mathbb{E}[\frac{N_1}{N_1+N_2}] = \frac{q_1}{q_1+q_2} (1-q_3^n)$, and
	 $\mathbb{E}[\frac{N_2}{N_1+N_2}] = \frac{q_2}{q_1+q_2} (1-q_3^n)$.
\end{lemma}

 \begin{proof} We only prove the first part and the proof of the second is similar.
	\begin{align}
& \mathbb{E}[\frac{N_1}{N_1+N_2}] \nonumber \\
&=  \sum_{n_1+n_2+n_3 = n,n_1+n_2 \neq 0}\frac{n_1}{n_1+ n_2} \binom{n}{n_3} \binom{n_1+n_2}{n_1}q_1^{n_1}q_2^{n_2}q_3^{n_3} \nonumber\\
	& = \sum_{n_1+n_2+n_3 = n,n_1 \neq 0}\frac{n_1}{n_1+ n_2} \binom{n}{n_3} \binom{n_1+n_2}{n_1}q_1^{n_1}q_2^{n_2}q_3^{n_3} \nonumber\\
	& = q_1\sum_{n_1+n_2+n_3 = n,n_1 \neq 0} \binom{n}{n_3} \binom{n_1+n_2-1}{n_1-1}q_1^{n_1-1}q_2^{n_2}q_3^{n_3} \nonumber\\
	& = q_1 \sum_{1\leq n_3 \leq n-1} \binom{n}{n_3}q_3^{n_3} [\sum_{n_1+n_2=n-n_3,n_1\neq 0}\binom{n_1+n_2-1}{n_1-1}q_1^{n_1-1}q_2^{n_2}]  \nonumber \\
	& = q_1 \sum_{1\leq n_3 \leq n-1} \binom{n}{n_3}q_3^{n_3} (q_1+q_2)^{n_1+n_2-1}\nonumber\\
	& = \frac{q_1}{q_1+ q_2}[(q_1+q_2+q_3)^n- q_3^n]\nonumber \\
	& = \frac{q_1}{q_1+ q_2}(1- q_3^n)\nonumber
		\end{align}
\end{proof}
\end{comment}

\begin{theorem} $\mathbb{E}[\frac{N_2 q- N_1p}{(N_1+N_2)(q-p)}| N_1+N_2 \neq 0] = \pi$.
\end{theorem}
%\begin{proof} It follows from the above two lemmas as well as the observation that $ Pr[N_1+N_2 = 0] = q_3^n$ where $Pr$ is the joint probability distribution of the sample.
%\end{proof}

\begin{theorem} \label{th:variance}
$Var(\hat{\pi}|N_1+N_2\neq0) = \frac{1}{(q-p)^2}[\pi p+(1-\pi)q][\pi q+(1-\pi)p]A = [-(\pi-\frac{1}{2})^2 + \frac{1}{4} (\frac{p+q}{p-q})^2] A$ where $A= \sum_{0\leq N_3 <n}\frac{1}{n-N_3}$ ${n}\choose{N_3}$$ (1-q_3)^{n-N_3}q_3^{N_3}$.

\end{theorem}

The formula in Theorem \ref{th:variance} is essential to our analysis of the trade-off between privacy loss and estimation accuracy. One may refer to the supplementary materials for a detailed proof (of independent interest). In this paper, we adopt  from \cite{GrabS54} a good approximation of $A$ as $\frac{1}{(n+1)(p+q)-1}$. In particular, with this approximation,   when $p+q=1$, $Var[\hat{\pi}| N_1+N_2 \neq 0] = \frac{-(\pi-\frac{1}{2})^2 + \frac{1}{4}\frac{1}{(2p-1)^2}}{n}$, which is  exactly the estimation error of Warner's model ( Eq. (\ref{Eq:variance})).

 \begin{corollary} Let $f(q)= \frac{-(\pi-\frac{1}{2})^2+ \frac{1}{4}(\frac{p+q}{p-q})^2}{(n+1)(p+q)-1}$. Then $f'(q) >0$. In other words, $Var(\hat{\pi})$ is increasing with respect to $q$.
\end{corollary}

This proposition tells us that, within the privacy budget of $\epsilon$, one can increase the estimation accuracy by saying ``I don't know"  as much as possible instead of lying.

\begin{corollary}
    Fix $p+q =c$. The optimal $\epsilon$-LDP mechanism is
    \begin{equation*}
Q_{GWRR}= \left(
\begin{array}{ccc}
\frac{e^{\epsilon}}{e^{\epsilon}+1} c & \frac{1}{e^{\epsilon}+1} c & 1-c \\
\frac{1}{e^{\epsilon}+1} c & \frac{e^{\epsilon}}{e^{\epsilon}+1} c  & 1-c
\end{array} \right)
\end{equation*}
\end{corollary}
%********************************

 In order to emphasize the dependency of the privacy matrix $Q_{2\times 3}$ on the parameters $p$ and $q$, we denote $Q_{2\times 3}$ as $Q_{2\times 3}(p,q)$, the privacy loss $ln(\frac{p}{q})$ as $\epsilon^S(p,q)$ and the estimation error $Var(\hat{\pi}|N_1+N_2\neq0)$ as $\nu^S(p,q)$.

This trade-off formula can be actually easily obtained. What we can achieve is an analysis rather than simulation. Let $p+q=c$ and $e^{\epsilon}= \frac{p}{q} =\frac{p}{1-c-p}$. So we get $p=\frac{1-c}{e^{-\epsilon}+1}$. If we substitute this formula into the error formula in Theorem \ref{th:variance}, then we get a formula of estimation error in terms of the privacy loss. 
Simulation experiments are carried out to verify the trade-off in the privacy mechanism. In order to reduce the sampling error on the experimental results, the following  results are the average of 1000 experimental outcomes.

\begin{figure}[htbp]
\centering
\includegraphics[height=5.0cm,width=6cm]{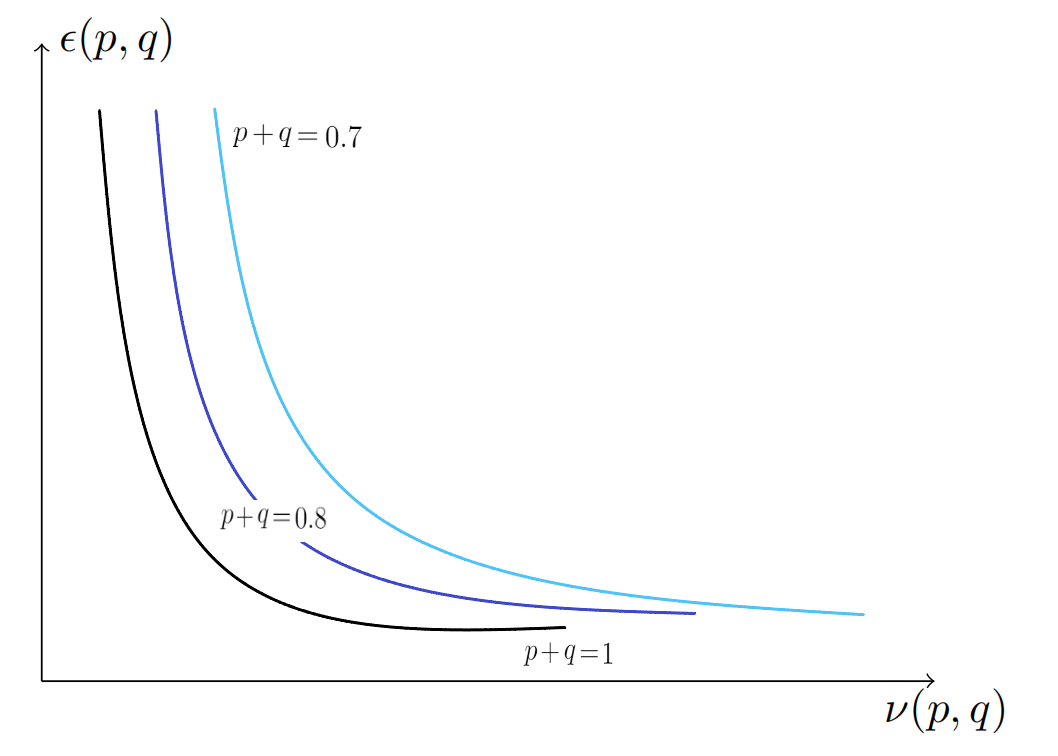}
\caption{The trade-off in Shafer's semantics}
\label{fig:2}
\end{figure} 

The trade-off between the privacy loss $\epsilon^S(p,q)$ and the accuracy $\nu^S(p,q)$ can be illustrated in the following Figure 1. The figure shows clearly the impact of ``don't know" with probability $1-p-q$ on the trade-off between $\epsilon^S(p,q)$ and $\nu^S(p,q)$.  When $1-p-q=0$ or $p+q=1$, the black curve for the trade-off between $\epsilon^S(p,q)$ and $\nu^S(p,q)$ is exactly for Warner's randomized response mechanism. If $p+q =c$ where $c$ is a constant, the trade-off curve is similar to that for Warner's mechanism. Moreover, when the constant $c$ gets smaller or the probability of ``don't know" gets larger,  the curve moves further away from that for Warner's model.
 Figure \ref{fig:2} tells us that Warner's model is optimal among those generalized $Q_{2\times 3}$-mechanisms. Next we explore the effect of the sample size on the accuracy of the estimation. We set the sample size to be 10, 100, 500, 1000 and fix $q_3 = 0.1$. From the experimental results (Figure \ref{fig:3}), we can see that when the privacy loss is relatively large, different sample sizes can achieve similar estimations. However, when the privacy budget is relatively small, with the increase of the sample size, the estimation variance gets smaller and smaller.

\begin{figure*}[htbp]
\centering
\subfigure[N=10]{
\begin{minipage}[t]{0.25\linewidth}
\centering
\includegraphics[height=3.9cm,width=3.9cm]{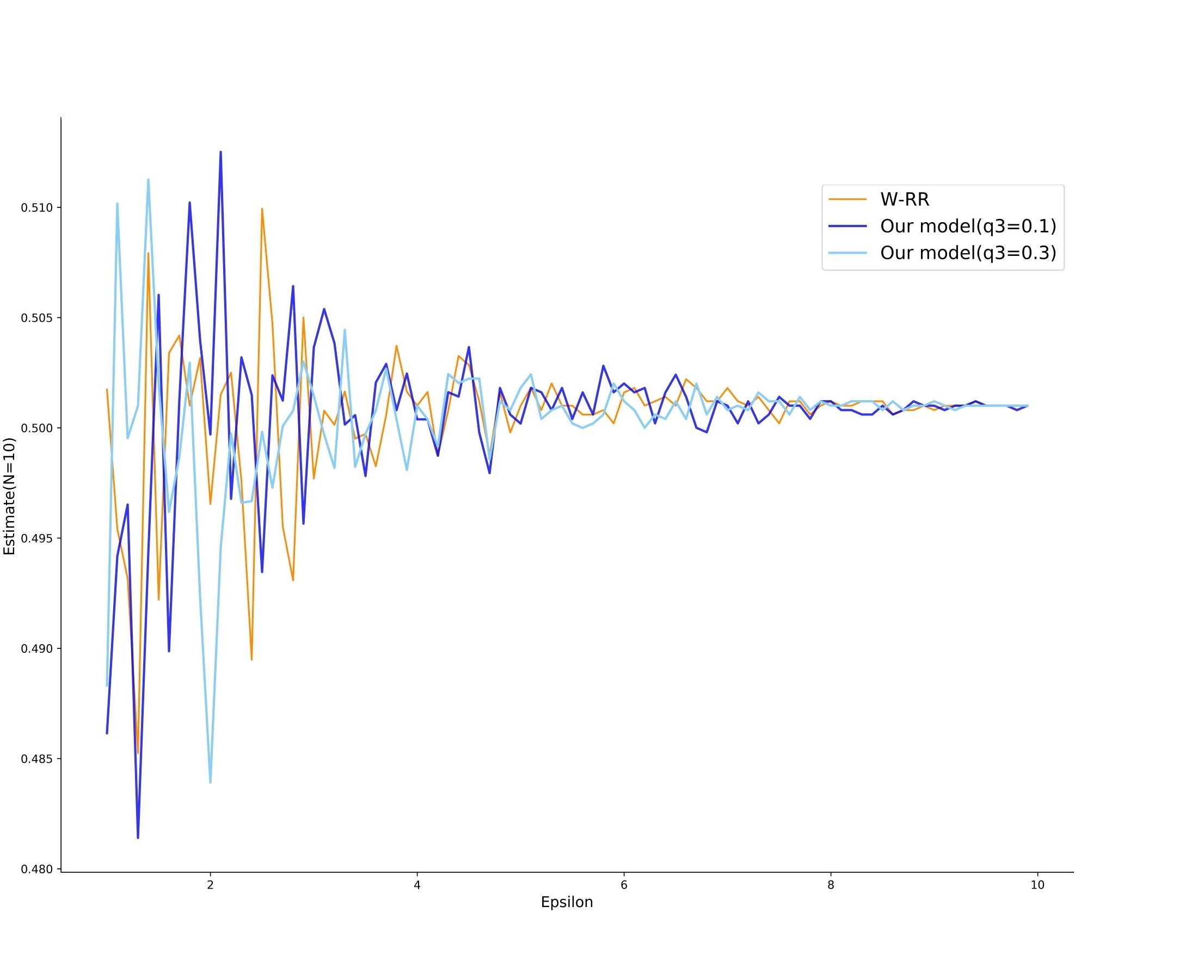}
%\caption{fig1}
\end{minipage}%
}%
\subfigure[N=100]{

\begin{minipage}[t]{0.25\linewidth}
\centering
\includegraphics[height=3.9cm,width=3.9cm]{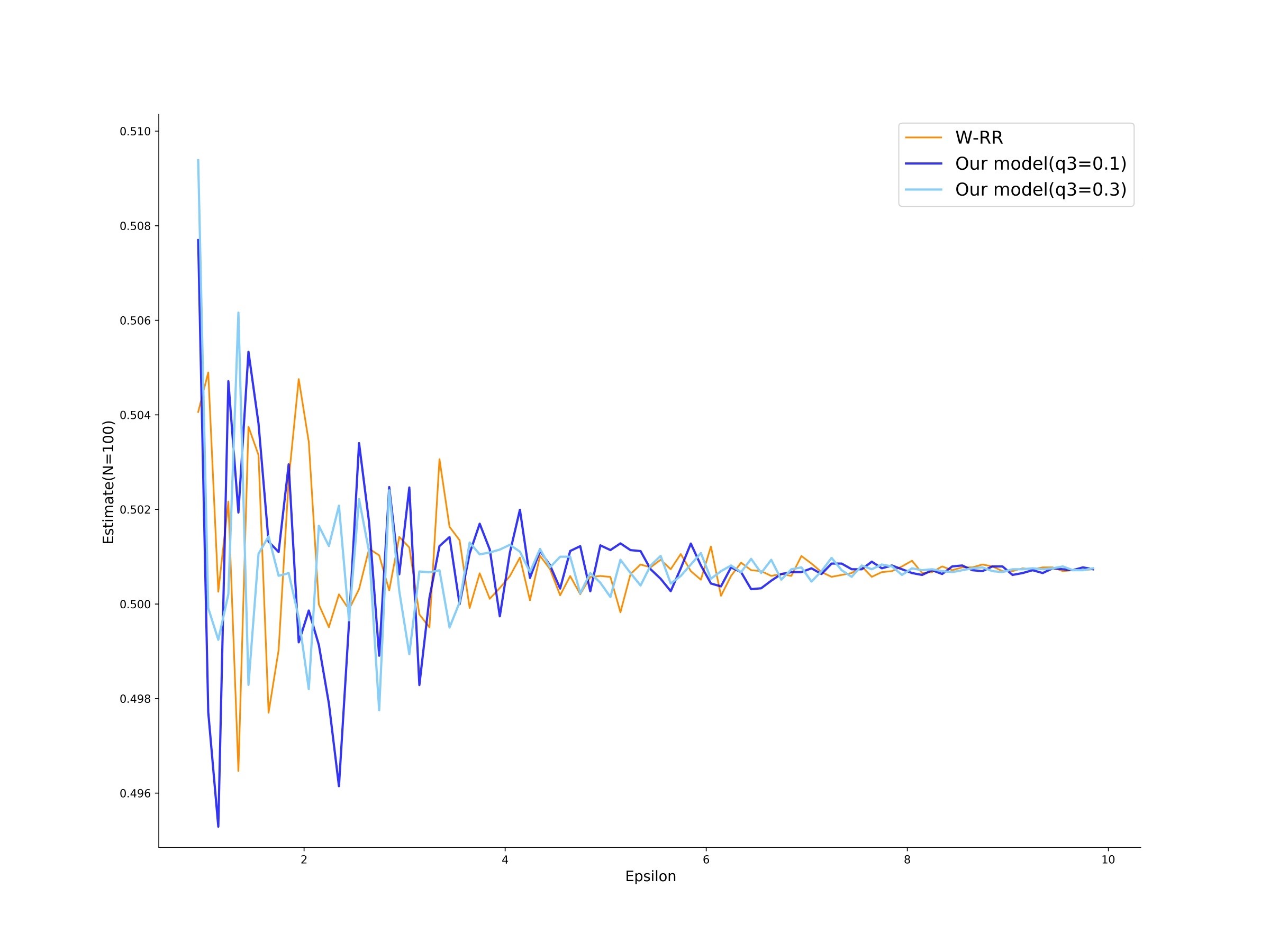}
%\caption{fig2}
\end{minipage}%
}%
\subfigure[N=500]{
\begin{minipage}[t]{0.25\linewidth}
\centering
\includegraphics[height=3.9cm,width=3.9cm]{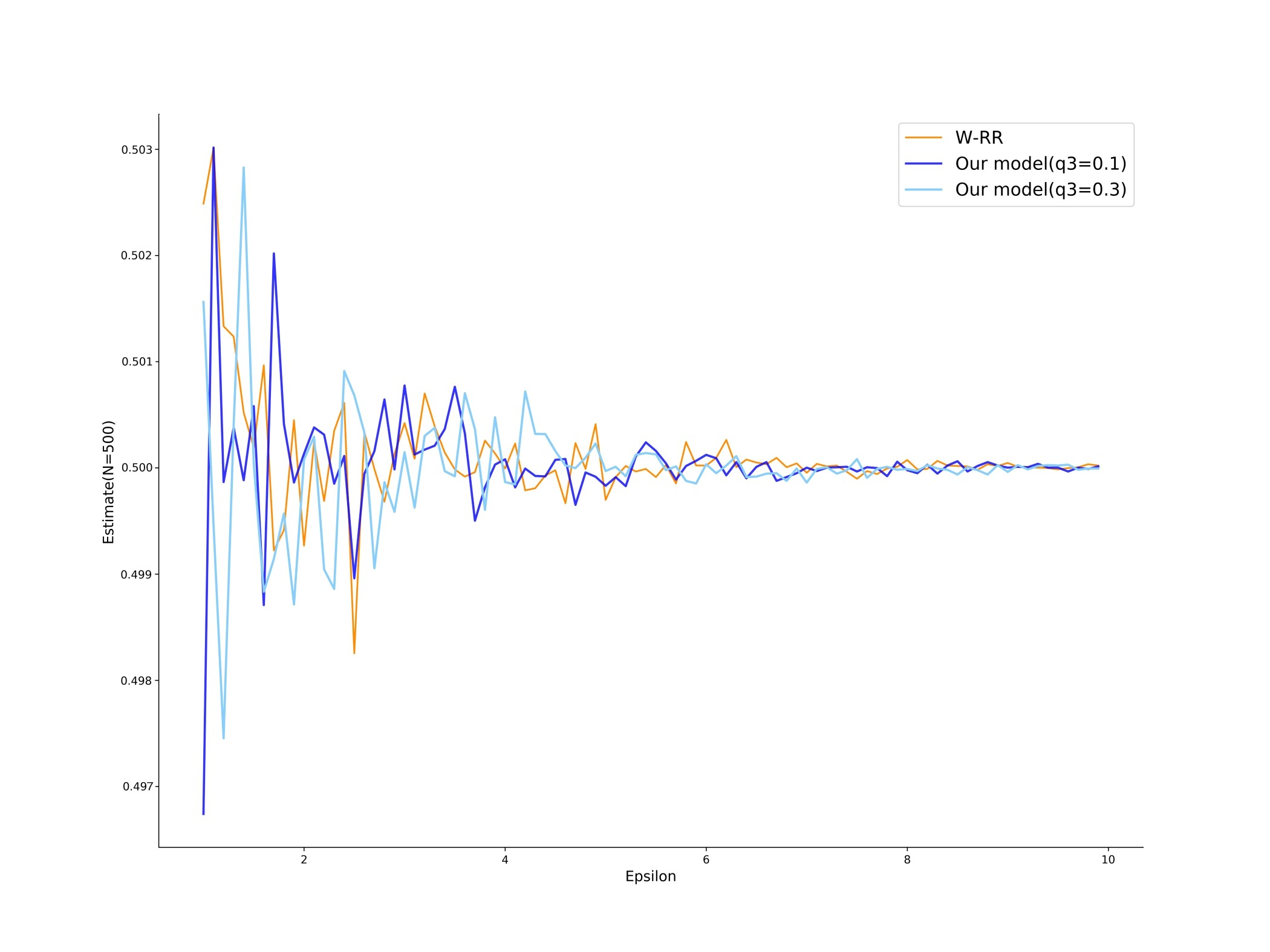}
%\caption{fig2}
\end{minipage}
}%
\subfigure[N=1000]{
\begin{minipage}[t]{0.25\linewidth}
\centering
\includegraphics[height=3.9cm,width=3.9cm]{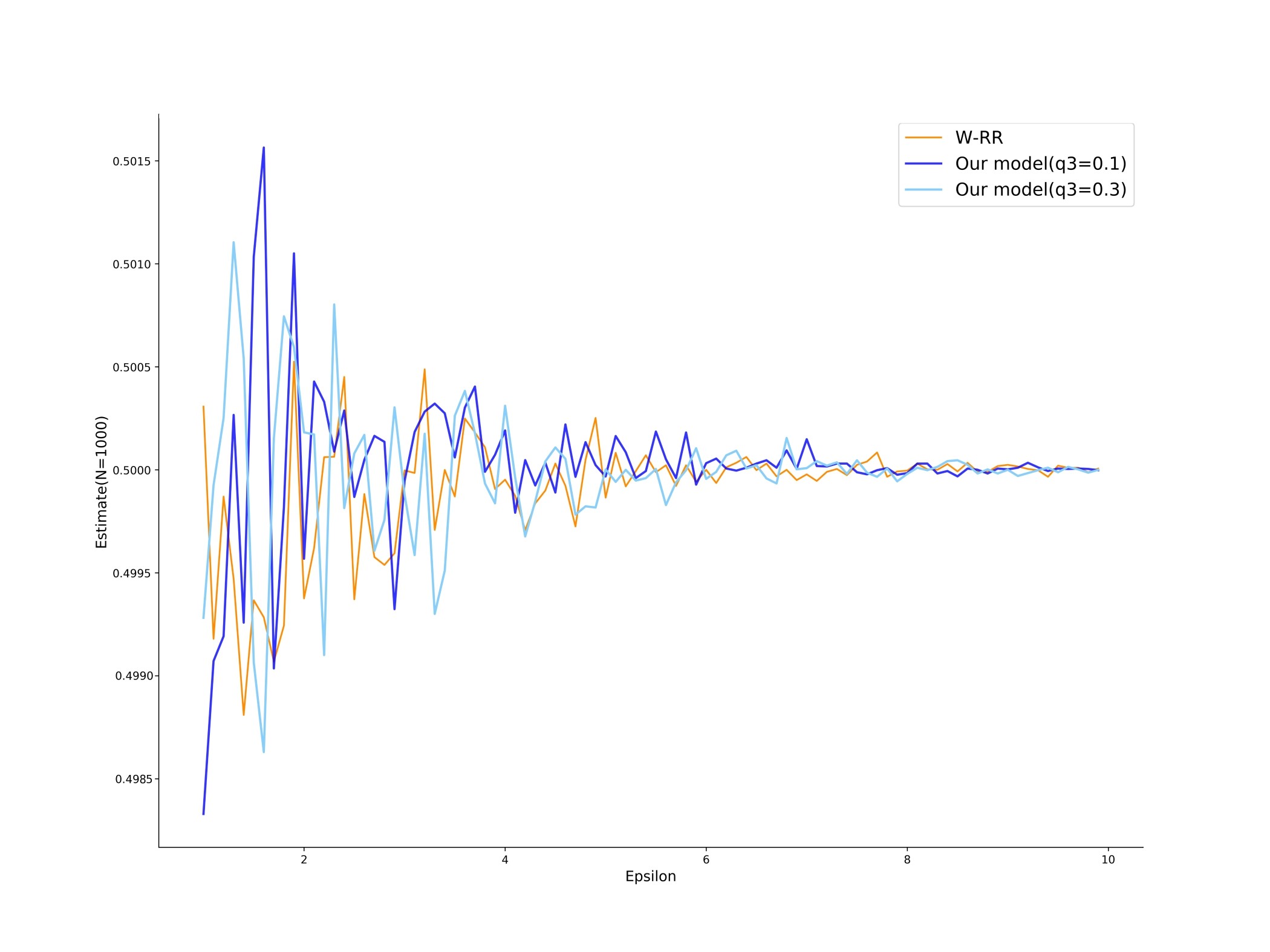}
%\caption{fig2}
\end{minipage}
}%

\centering
\caption{ Impact of sample sizes on the estimation accuracy}
\label{fig:3}
\end{figure*}

%We first explored the impact of different probability of responses "I don't know" on the estimates. As can be %seen from the figure, with the increase of privacy budget, the estimated value will be close to the real value; %but with the increase of $q_3$, the variance of the estimated value will increase. Note: Warner's model is a %special case of our model with $q_3=0$.

%********************************************************************
\subsection{LDP according to Walley}

For an evidential privacy mechanism $Q$, let
$r_Q^W = max_{pr_x\in \mathcal{P}_{bel^Q_x},pr_{x'}\in \mathcal{P}_{bel^Q_{x'}} }\frac{pr_x(E)}{pr_{x'}(E)}$.
And the logarithm $\epsilon_Q^W = ln(r_Q^W)$ quantifies the privacy loss of the privacy mechanism $Q$ in Walley's semantics of imprecise probabilities.
There is another definition of LDP for belief functions in the setting of imprecise probabilities:

\begin{definition} \label{def:WLDP} For any $\epsilon >0$,
$Q$ is called \emph{$\epsilon$-locally differential private} according to Walley ($\epsilon$-$WLDP$ for short) if, 
$ -\epsilon \leq \epsilon_Q^W \leq \epsilon$.
 And $\epsilon^Q_W$ is called the \emph{privacy loss} of $Q$ according to Walley and $\epsilon$ is a \emph{privacy budget}.
\end{definition}
In other words, the privacy loss for $\epsilon$-$WLDP$ is defined by consistent probability functions \emph{in the worst case}. So, $\epsilon$-$WLDP$ fits well with the worst-case analysis behind the philosophy of differential privacy and also with the \emph{conservative} principle of least commitment in the theory of belief functions \cite{Denoeux14}.  Lemma \ref{le:alternation1} and the following Lemma \ref{le:charaterization_Walley} provide a simple mathematical characterization of SLDP and WLDP, where we can see clearly the main difference between Definitions \ref{def:SLDP} and \ref{def:WLDP}.

\begin{lemma} \label{le:charaterization_Walley}(Alternative formulations) If privacy mechanism $Q$ is $\epsilon$-$WLDP$, then, for all $x, x'\in X$ and $E\subseteq Y$:
 $e^{-\epsilon} \leq \frac{pl_x(E)}{bel_{x'}(E)} \leq e^{\epsilon}$.
\end{lemma}

\begin{lemma} \label{le:composition}(Composition) Let $Q_1$ be an $\epsilon_1$-WLDP  evidential privacy mechanism from $X$ to $Y_1$ and $Q_2$ be an $\epsilon_2$-WLDP  evidential privacy mechanisms from $X$ to $Y_2$.  Then their combination $Q_{1,2}$ defined by $Q_{1,2}(x) = (Q_1(x),Q_2(x) )$ is $\epsilon_1 + \epsilon_2$-WLDP.
\end{lemma}

\begin{lemma} \label{le:post}(Post-processing) Let $Q$ be an $\epsilon$-WLDP  mechanism from $X$ to $Y$ and $f$ is a data-independent randomized algorithm from $Y$ to another finite alphabet set $Z$. Then $f\circ Q$ is an $\epsilon$-WLDP mechanism from $X$ to $Z$.
\end{lemma}

%\begin{lemma}
%(Simple composition) Let $Q_1$ be an $\epsilon_1$-WLDP from $X_1$ to $X_2$ and $Q_2$ be an $\epsilon_2$-WLDP from $X_2$ to $X_3$.  Then the simple composition $Q_2\circ Q_1$ from $X_1$ to $X_3$ is an $\epsilon_1 + \epsilon_2$-WLDP.
%\end{lemma}
%[]ADD Proof]

For the hypothesis testing problem, recall that $Q$ denotes an evidential privacy mechanism and $\phi: Y \rightarrow [0,1]$ is a rejection rule. In order to translate $\epsilon$-WLDP into the trade-off between type \RomanNumeralCaps{1} and \RomanNumeralCaps{2} errors, we have to divide them into two different types of errors: one is pessimistic and the other optimistic. For the rejection rule $\phi$, the \emph{pessimistic} type  \RomanNumeralCaps{1} and \RomanNumeralCaps{2} are defined as $\alpha^{pe}_{\phi} = \text{sup}_{pr\in \mathcal{P}_{bel_x^Q}}\mathbb{E}_{pr} (\phi)$ and $\beta^{pe}_{\phi} = \text{sup}_{pr\in \mathcal{P}_{bel_{x'}^Q}}\mathbb{E}_{pr} (1-\phi)$, respectively. They are actually the same as those errors in $\epsilon$-$SLDP$. Also we define the \emph{optimistic} type  \RomanNumeralCaps{1} and \RomanNumeralCaps{2} errors as $\alpha^{op}_{\phi}:= \text{inf}_{pr\in \mathcal{P}_{bel_x^Q}}\mathbb{E}_{pr} (\phi)$ and $\beta^{op}_{\phi}: = \text{inf}_{pr\in \mathcal{P}_{bel_{x'}^Q}}\mathbb{E}_{pr} (1-\phi)$, respectively. 

\begin{definition} For the above pessimistic errors, the following function is called the \emph{pessimistic trade-off function}:
			$T^{pe}(Q(x), Q(x'))(\alpha):= inf\{\beta^{pe}_{\phi}: \alpha^{pe}_{\phi}\leq \alpha\}$.
For the above optimistic errors, the following function is called the \emph{optimistic trade-off function}:
			$T^{op}(Q(x), Q(x'))(\alpha):= sup\{\beta^{op}_{\phi}: \alpha^{op}_{\phi}\leq \alpha\}$.
	\end{definition}

 The following theorem is another generalization of the well-known result (Theorem 2.4 in \cite{WassermanZ2010}) for standard differential privacy. 

\begin{theorem} \label{thm:equivalence4WLDP}For any evidential privacy mechanism $Q$, the following two statements are equivalent:
	\begin{enumerate}
		\item $Q$ is $\epsilon$-WLDP;
		\item For any $\alpha\in [0,1]$, $T^{pe}(Q(x), Q(x')) (\alpha)\geq f^{pe}_{\epsilon}(\alpha)$ and $T^{op}(Q(x), Q(x'))(\alpha)\leq f^{op}_{\epsilon}(\alpha)$ where $f^{pe}_{\epsilon}(\alpha)= max\{1-\alpha e^{\epsilon}, 0, e^{-\epsilon}(1-\alpha)\}$ and $f^{op}_{\epsilon}(\alpha)= min\{1-\alpha e^{-\epsilon}, e^{\epsilon}(1-\alpha)\}$.
	\end{enumerate}
\end{theorem}

For the composition, the adversary needs to distinguish between $Q(x)\times Q(x)$ and $Q(x')\times Q(x')$. Similarly, we can define pessimistic and optimistic type \RomanNumeralCaps{1} and \RomanNumeralCaps{2} errors: $\alpha^{2,pe}_{\phi},\beta^{2,pe}_{\phi}, \alpha^{2,op}_{\phi}$ and $\beta^{2,op}_{\phi}$. Moreover, for the hypothesis testing problem for the composition, we define the pessimistic and optimistic trade-off functions similarly: $T^{pe}_2(Q(x)\times Q(x), Q(x')\times Q(x'))(\alpha):= inf\{\beta^{2,pe}_{\phi}: \alpha^{2,pe}_{\phi}\leq \alpha\}$, and 
	 $T^{op}_2(Q(x)\times Q(x), Q(x')\times Q(x'))(\alpha):= sup\{\beta^{2,op}_{\phi}: \alpha^{2,op}_{\phi}\leq \alpha\}$.
\

\begin{corollary}\label{cor:composition4WLDP} For any $\alpha\in [0,1]$, $T_2^{pe}(Q(x)\times Q(x), Q(x')\times Q(x')) (\alpha)\geq f^{2,pe}_{\epsilon}(\alpha)$ and $T_2^{op}(Q(x)\times Q(x), Q(x')\times Q(x'))(\alpha)\leq f^{2,op}_{\epsilon}(\alpha)$ where $f^{2,pe}_{\epsilon}(\alpha)= max\{1-\alpha e^{2\epsilon}, -\alpha + \frac{2}{e^{\epsilon}+1}, e^{-2\epsilon}(1-\alpha)\}$ and $f^{2,op}_{\epsilon}(\alpha)= min\{1-\alpha e^{-2\epsilon}, e^{2\epsilon}(1-\alpha), -\alpha + \frac{3-e^{-2\epsilon}}{e^{\epsilon}+1}\}$.
\end{corollary}

Both Theorem \ref{thm:equivalence4WLDP} and Corollary \ref{cor:composition4WLDP} can be visualized as follows:

\begin{figure}[htbp]
\centering
\includegraphics[height=5.0cm,width=6cm]{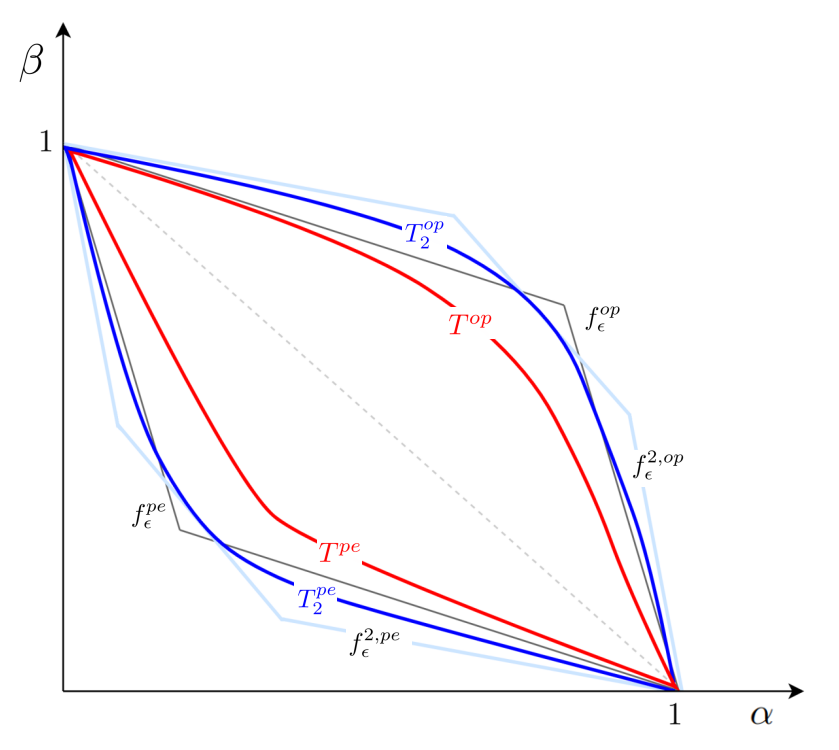}
\caption{The trade-off between  type \RomanNumeralCaps{1} and \RomanNumeralCaps{2} errors for WLDP}
\label{fig:wldp}
\end{figure}

For simplicity, we consider the above  evidential privacy matrix
\begin{equation*}
Q_{2\times 3}= \left(
\begin{array}{ccc}
p & q & 1-p-q \\
q & p & 1-p-q
\end{array} \right).
\end{equation*}
  In Definition \ref{def:SLDP}, $1-p-q$ quantifies the conditional probability of the third response ``I don't know".  Similarly, in Definition \ref{def:WLDP}, $p$ and $q$ are the probabilities of telling truthfully and of lying respectively. However, $1-p-q$  measures the probability of \emph{unknown} response strategy or \emph{possible} noncompliance. Unlike SLDP, there are only two responses ``Yes" and ``No" for response mechanism according to WLDP and ``I don't know" is not an option.  In order to obtain a Warner-style randomized response $2\times 2$ matrix, we redistribute the mass $1-p-q$ on the unknown part to those masses on ``Yes" and ``No"  and get the following matrix:

\begin{equation*}
Q_{\lambda}= \left(
\begin{array}{ccc}
p +  \lambda(1-p-q) & q + (1-\lambda)(1-p-q)  \\
q + (1-\lambda)(1-p-q) & p +\lambda(1-p-q)
\end{array} \right)
\end{equation*}
 When $\lambda=1$, the associated privacy loss is the largest and is the same as according to Definition \ref{def:WLDP}. The respondent is most conservative and make the worst-case analysis.   On the other hand, when $\lambda=0$, the associated privacy loss is the smallest. In this case, the respondent is the most optimistic and assumes the best possibility.  Similarly,  we can obtain the maximum likelihood estimation
$\hat{\pi} = \frac{\frac{n_1}{n}-(1-\lambda)(1-p-q)-q}{p-q+(2\lambda-)(1-p-q)}$, and show that
$\hat{\pi}$ is an unbiased estimate of $\pi$.  From Theorem \ref{th:variance}, we know that, when $\lambda =0$, the variance $Var(\hat{\pi}) (= \frac{-(\pi-1/2)^2 + \frac{1}{4(2p-1)^2}}{n})$ is the largest and is defined as \emph{the estimation accuracy} of the privacy matrix $Q_{2\times 3}$ according to Walley.

According to Shafer's semantics, the privacy loss for the mechanism $Q_{2\times 3}$ is defined as  $\epsilon^S(p,q)= ln(\frac{p}{q})$ and its  accuracy is $\nu^S(p,q)=Var(\hat{\pi} |N_1+ N_2 \neq 0) = \frac{-(\pi-\frac{1}{2})^2 + \frac{(p-q)^2}{4(p+q)^2}}{(n+1)(p+q)-1}$ (Thm. (\ref{th:variance})).  In contrast, according to Walley's semantics, the privacy loss for $Q_{2\times 3}$ is defined as $ln(\frac{1-q}{q})$, which is denoted as $\epsilon^W(p,q)$ and is equal to the privacy loss of the associated matrix $Q_1$ in Warner's model.   Moreover its accuracy is $\frac{-(\pi-\frac{1}{2})^2 + \frac{1}{4(2p-1)^2}}{n}$, which is denoted as $\nu^W(p, q)$ and is exactly the accuracy for the matrix $Q_0$ in Warner's model.  In other words, both $\epsilon^W(p, q)$ and $\nu^W(p,q)$ are obtained according to the worst-case analysis from the perspectives of the respondent and adversary respectively.  Similarly, we may obtain $\epsilon^O(p, q)$ and $\nu^O(p, q)$, the optimal privacy loss and estimation error among all possible privacy mechanisms $Q_{\lambda}$.  The following Figure \ref{fig:tradeoffs} illustrates the relationships among the three trade-offs between privacy and accuracy: $(\epsilon^S(p, q), \nu^S(p, q)), (\epsilon^W(p, q), \nu^W(p, q))$ and $(\epsilon^O(p, q), \nu^O(p, q))$.  The rectangle shown in the figure consists of exactly the trade-offs between privacy and accuracy for all possible $Q_{\lambda}$ with $(\epsilon^W(p,q),\nu^W(p,q))$ as the worst and $(\epsilon^O(p,q),\nu^O(p,q))$ as the best.

% \begin{figure}
%     \centering
%     \begin{tikzpicture}
%         \draw[->](0,0)--(7,0)node[left,below]{$\epsilon(p,q)$};
%         \draw[->](0,0)--(0,5)node[right]{$\nu(p,q)$};
%         \draw[dashed](1,1)--(1,4);
%         \draw[dashed](1,4)--(6,4);
%         \draw[dashed](6,1)--(6,4);
%         \draw[dashed](1,1)--(6,1);
%         \draw[fill] (1,1) circle (.03);
%         \draw[fill] (6,4) circle (.03);
%         \draw[fill] (2.5,3) circle (.03);
%         \path (1.5,1) node [below] {$(\epsilon^O(p,q),\nu^O(p,q))$};
%         \path (6,4) node [above] {$(\epsilon^W(p,q),\nu^W(p,q))$};
%         \path (3,3) node [below] {$(\epsilon^S(p,q),\nu^S(p,q))$};
%     \end{tikzpicture}
%     \caption*{Picture}
%     \label{pic:nuepsilon}
% \end{figure}

\begin{figure}[htbp]
\centering
\includegraphics[height=4.0cm,width=6cm]{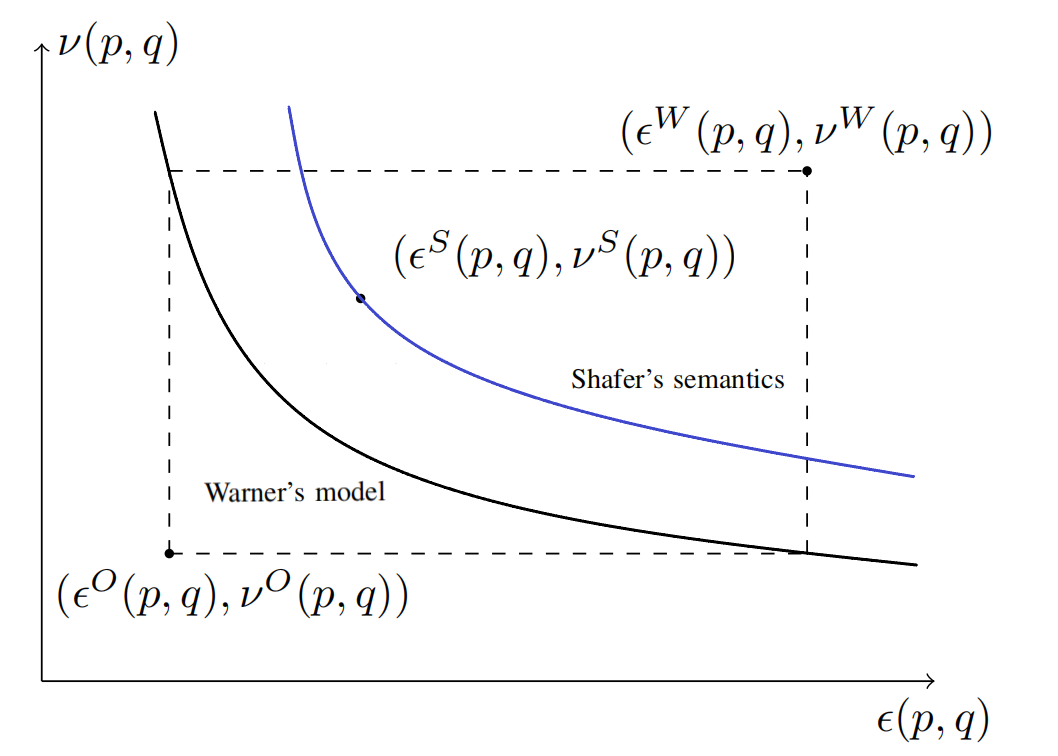}
\caption{Comparison of trade-offs in the two semantics}
\label{fig:tradeoffs}
\end{figure}

\begin{corollary} $\epsilon^W(p,q)$ is decreasing with respect to $q$ and $\nu^W(p,q)$ is decreasing with respect to $p$.
\end{corollary}

According to the corollary, we may compare two privacy mechanisms $Q_{2\times 3}(p, q)$ and $Q_{2\times 3}(p', q')$. If $p\geq p'$ and $q\geq q'$, then $\epsilon^W(p, q) \leq \epsilon^W(p', q')$ and $\nu^W(p, q) \leq \nu^W(p', q')$. In this case, $Q_{2\times 3}(p, q)$ is \emph{preferred} to $Q_{2\times 3}(p', q')$.  So the trade-off in Walley's semantics is similar to the minimax estimation for LDP \cite{DuchiJW18}.

\begin{comment}
\begin{corollary} For any privacy mechanism $Q_{2\times 3}(p,q)$,  there is another privacy mechanism $Q_{2\times 3}(p', q')$ such that
	\begin{enumerate}
		\item $\epsilon^{S}(p', q') = \epsilon^W(p, q)$; and
		\item $\nu^S(p', q') = \nu^W(p, q)$.
	\end{enumerate}
\end{corollary}
In other words, the trade-off for Walley's semantics can be translated into that for Shafer's semantics.
\end{comment}

%&&&&&&&&&&&&&&&&&&&&&&&&&&&&&&&&&&&&&&&&&

%&&&&&&&&&&&&&&&&&&&&&&&&&&&&&&&&&&&&&&&77

\section{Conclusion}
%\bigskip
\noindent To the best of our knowledge, we are the \emph{first} to explore differential privacy from a different uncertainty perspective than probability theory.  The fact that differential privacy is closely related to statistical  analysis \cite{DworkR14} may explain why there are few research about DP in other uncertainty theories which don't support a practical statistical analysis. But belief functions are deeply rooted in fiducial inference,  an important school in statistics \cite{Dempster67,Shafer82, MartinL15,Martin2019false}. It is desirable to develop a \emph{belief-function} theory of differential privacy. The LDP implicitly requires some assumptions about the adversary's view of belief functions in privacy mechanism. There are many semantics for belief functions. In this paper, we choose Shafer's semantics as randomly encoded messages \cite{ShaferT85} and Walley's  interpretation as imprecise-probabilities \cite{Walley90}. Our work in LDP is motivated by the nonresponse and noncompliance issue in randomized response technique in \cite{Warner65,GraemeIZ15} and discrete distribution estimation problem in \cite{KairouzOV16,KairouzBR16,WangBLJ17,HuangD2008} where the size of the input alphabet is no less than that of the output alphabet. However, since the number of messages (or the size of the powerset of the output set) is usually larger than that of the input set in our LDP mechanisms, MLE is usually different from empirical estimation in this case and their techniques don't apply here. Moreover, there is a rich literature to address nonresponse in survey research \cite{LittleR2002} but most of them regard the issue   as a missing-data problem and few of them consider the privacy problem. There seems no obvious LDP definitions for coarsening at random because the outputs of coarsening mechanisms at different inputs  are different and hence the adversary can easily distinguish these two inputs. It may be interesting to explore the LDPs for contamination models. There are  2 other possible definitions of SLDP in terms of belief functions and plausibility functions: $e^{-\epsilon} \leq \frac{bel_x^Q(E)}{bel_{x'}^Q(E)}\leq e^{\epsilon}$ and $e^{-\epsilon} \leq \frac{pl_x^Q(E)}{pl_{x'}^Q(E)}\leq e^{\epsilon}$. Lemma \ref{le:alternation1} and the remarks afterwards actually show their relationships. In future versions, we will elaborate these two different definitions and their relations with Definition \ref{def:SLDP}.

In this paper we show a  binary composition theorem for each definition (Corollaries \ref{cor:composition4SLDP} and \ref{cor:composition4WLDP}). We believe that, for our two definitions SLDP and WLDP, the composition of the hypothesis-testing trade-off functions \cite{Kairouz2017composition,Balle2020hypothesis}  converges to some (most probably random-set variant) form of Gaussian DP  \cite{DongRS21} according to some central limit theorem (Chapter 3 in \cite{Molchanov17}). In this paper, we took the first step in this direction and showed the effect of the composition of hypothesis-testing trade-off functions(Corollaries 1 and 4).  Moreover, we would like to investigate LDP for belief functions from the perspective of respondents (as in \cite{Xiong0LJ20}) and  conduct a series of rigorous surveys to show that our new generalized Warner's mechanism including ``don't know" as an option can indeed increase user's willingness to participate.

\section*{Acknowledgements}
The corresponding author wants to thank Professors Arthur Dempster, Xiao-Li Meng and Ruobin Gong for their support during his visiting scholarship at Harvard Statistics Department when this research was initiated. The definition of LDP according to Walley was inspired by an insightful discussion with Professor Xiao-Li Meng. The research is partly supported by NSFC (61732006) and the third author is supported by NSFC (No.61772534).

\begin{comment}

 For a given evidential privacy mechanism $Q$, let
\begin{align}
    \epsilon^Q_S &= max_{x,x'\in X, E\subseteq Y}\frac{m_x(E)}{m_{x'}(E)}\\
    \epsilon^Q_W &= max_{pr_x\in \mathcal{P}_{bel^Q_x},pr_{x'}\in \mathcal{P}_{bel^Q_{x'}},E\subseteq Y }\frac{pr_x(E)}{pr_{x'}(E)}
\end{align}
Let $Q_W$ be a $k\times l$ row-stochastic matrix that satisfies the following conditions:
\begin{enumerate}
    \item each row $m^{Q_W}_x$ is a probability function on $Y$ which is consistent with $m^Q_x$;
    \item $\epsilon^Q_S = \epsilon^{Q_W}_S$
\end{enumerate}
Let $\mathcal{Q}_W$ denote the set of all these $k\times l$ row-stochastic matrices. In order to compare privacy-preserving in these two \emph{different} kinds of local differential privacy for the same $Q$, we only need to compare $\epsilon^Q_S$ and $\epsilon^Q_W$.  In order to compare utilities of these two local differential privacy, we consider their accuracy in discrete distribution estimation.

\end{comment}
 \bibliographystyle{plain}
   \bibliography{Savage}
\newpage
\begin{large}
   \textbf{Supplementary-Materials}
\end{large}

\vspace*{3 ex}

\textbf{Proof of Lemma 3.2}:  These two propositions follow directly from the inequality (1) and the facts that $bel_x(E) = \sum_{E'\subseteq E} m(E')$ and $pl_x(E) = \sum_{E'\cap E \neq \emptyset} m(E')$.

\textbf{Proof of Lemma 3.3}: The composition follows from the fact that $\frac{m^{Q_{12}}_x(E_1, E_2)}{m^{Q_{12}}_{x'}(E_1, E_2)}=\frac{m^{Q_{1}}_x(E_1)}{m^{Q_{1}}_{x'}(E_1)}\frac{m^{Q_{2}}_{x}(E_2)}{m^{Q_{2}}_{x'}(E_2)}$.

\textbf{Proof of Lemma 3.4}: The proof for $\epsilon$-SLDP algorithm is similar to that of the post-processing property for probabilistic privacy mechanisms (Proposition 2.1 in \cite{DworkR14}). We prove the proposition only for a deterministic function $f: Y\rightarrow Z$. For any $x, x'\in X$, $m_x^{f\circ Q} (E) = m_x^Q (f^{-1}(E)) \leq e^{\epsilon} m_{x'}^Q (f^{-1}(E))$.

\textbf{Proof of Theorem 3.5}:  Let's assume that the rejection region for the hypothesis testing problem  is $R$. Let  $Q$ be $\epsilon$-SLDP. 
	\begin{enumerate}
		\item  Assume that $pl_x^{Q}(R)\leq \alpha$. It follows that, for any $E\subseteq Y$, $-\epsilon \leq \frac{m_x^Q}{m^Q_{x'}}\leq e^{\epsilon}$. We have that $e^{\epsilon} bel_{x'}^Q(R^c) \geq bel_x^Q(R^c) \geq 1-\alpha$. So $bel_{x'}^Q(R^c) \geq e^{-\epsilon} (1-\alpha)$. So, $pl_{x'}(R) \leq 1-e^{-\epsilon}(1-\alpha)$.  Moreover, $pl_{x'}^Q(R) \leq e^{\epsilon}pl_x^{Q}(R) = e^{\epsilon} \alpha $.  By putting these together, we obtain that $pl_{x'}^Q(R) \leq min\{e^{\epsilon} \alpha,1-e^{-\epsilon}(1-\alpha)\}$. This implies that the type 2 error $\beta_{\phi} \geq max\{1-e^{\epsilon} \alpha,e^{-\epsilon}(1-\alpha)\}$.  
	\item If we asssume that, $bel_{x}^Q(R)\leq \alpha'$, then $pl_{x}^Q(R^c)\leq 1-\alpha'$. Since $e^{-\epsilon}pl_{x'}^Q(R^c)\leq pl_x^Q(R^c)$, $pl_{x'}^Q(R^c)\leq e^{\epsilon}(1-\alpha')$. It follows that $bel_{x'}^Q(R)\geq 1-e^{\epsilon} (1-\alpha')$.  Moreover, $bel_{x'}^Q(R) \geq e^{-\epsilon}bel_x^Q(R) \geq \alpha' e^{-\epsilon}$.  By puting these together, we have that $bel_{x'}^Q(R) \geq max\{\alpha' e^{-\epsilon}, 1-e^{\epsilon}(1-\alpha')\}$.
	\end{enumerate}

In other words, if type  \RomanNumeralCaps{1} error $\alpha_{\phi}\in [l, L]$, then type \RomanNumeralCaps{2} error 
$\beta_{\phi}\in [u(L), U(l)]$ where $u(\alpha) = max\{e^{-\epsilon}(1-\alpha), 1-\alpha e^{\epsilon}\}$ and $U(\alpha) = min\{e^{\epsilon}(1-\alpha),1-\alpha e^{-\epsilon}\}$.  Since the reasonings above are birectional, the direction from (2) to (1) can be shown similarly.

\textbf{Proof of Corollary 3.6}: 
Let  $Q_{2\times 3}$ denote a known row-stochastic matrix as follows:

\begin{equation*}
Q_{2\times 3} = \left(
\begin{array}{ccc}
p & q & 1-p-q \\
q & p & 1-p-q
\end{array} \right)
\end{equation*}
where $p, q\in [0,1]$. $Q_{2\times 3}$ may be regarded as a generalized Warner's randomized response mechanism where a respondent answers truthfully with probability $p$, tells a lie with $q$ and don't respond or respond ``I don't know" with probability $1-p-q$.  Let $Q_x$ denotes the first row of $Q$ and $Q_{x'}$ the second row.  Now we consider the hypothesis testing problem by distinguishing between $Q_x\times Q_x$ and $Q_{x'}\times Q_{x'}$. Let $m^{\otimes 2}_x$ and $m^{\otimes 2}_{x'}$ denote their corresponding mass functions.  Their output set is $\{0,1,2\}$.

\begin{center}
\begin{tabular}{  c | c | c | c | c | c  | c }
$m^{\otimes 2}_x$ & $\{0\}$ & $\{1\}$ & $\{2\}$ & $\{0,1\}$ & $\{1,2\}$ & $\{0,1,2\}$\\
\hline 
& $p^2$ & $2pq$ & $q^2$ & $2p(1-p-q)$ & $2q(1-p-q)$ &   $(1-p-q)^2 $
\end{tabular}
\end{center}

And 
\begin{center}
\begin{tabular}{  c | c | c | c | c | c  | c }
$m^{\otimes 2}_{x'}$ & $\{0\}$ & $\{1\}$ & $\{2\}$ & $\{0,1\}$ & $\{1,2\}$ & $\{0,1,2\}$\\
\hline 
& $q^2$ & $2pq$ & $p^2$ & $2q(1-p-q)$ & $2p(1-p-q)$ &   $(1-p-q)^2 $
\end{tabular}
\end{center}

Since $Q$ is $\epsilon$-SLDP, we have 
\begin{center}
	$e^{-\epsilon} \leq \frac{1-q}{q}, \frac{1-p}{p} \leq e^{\epsilon}$.   
\end{center}

 Let $R$ be a rejection region and assume that $\alpha' \leq pl_{x}^{Q\times Q} (R) \leq \alpha$. If $(1-p)^2 \leq \alpha$, then $p^2\leq e^{2\epsilon}\alpha$.  Since $e^{\epsilon}\geq 1$, $p \geq \frac{2-e^{\epsilon}}{e^{\epsilon}+1}$.  It follows that $(2-p)p \leq e^{2\epsilon}(1-p^2)$ and $1-(1-p)^2 \leq e^{2\epsilon}(1-p^2)$. So we have that $p^2 \leq 1-e^{-2\epsilon}[1-(1-p)^2]$. Since $1-e^{-2\epsilon}[1-(1-p)^2] \leq 1 - (1-x) e^{-2\epsilon}$.  so we havs shown that if the type 1 error is $\leq \alpha$, then the type 2 error is at least $u^2(\alpha): = max\{e^{-2\epsilon}(1-\alpha), -\alpha+\frac{2}{e^{\epsilon}+1}, 1-\alpha e^{2\epsilon}\}$.  For the other part, we can show similarly: if the type 1 error is at most $\alpha$, then the type 2 error is at least $U^2(\alpha): = min\{e^{2\epsilon}(1-\alpha),1-\alpha e^{-2\epsilon}, -\alpha +\frac{3-e^{-2\epsilon}}{e^{\epsilon}+1}\}$. 

%\begin{lemma}
%\begin{itemize}%
%	\item $\mathbb{E}[\frac{N_1}{N_1+N_2}] = \frac{q_1}{q_1+q_2} (1-q_3^n)$;%
%	\item $\mathbb{E}[\frac{N_2}{N_1+N_2}] = \frac{q_2}{q_1+q_2} (1-q_3^n)$.
%\end{itemize}
%\end{lemma}

\begin{lemma}
	$\mathbb{E}[\frac{N_1}{N_1+N_2}] = \frac{q_1}{q_1+q_2} (1-q_3^n)$, and
	 $\mathbb{E}[\frac{N_2}{N_1+N_2}] = \frac{q_2}{q_1+q_2} (1-q_3^n)$.
\end{lemma}

\textbf{Proof of Lemma 0.1}:  We only prove the first part and the proof of the second is similar. 
	\begin{align}
	\mathbb{E}[\frac{N_1}{N_1+N_2}] &=  \sum_{n_1+n_2+n_3 = n,n_1+n_2 \neq 0}\frac{n_1}{n_1+ n_2} \binom{n}{n_3} \binom{n_1+n_2}{n_1}q_1^{n_1}q_2^{n_2}q_3^{n_3} \nonumber\\
	& = \sum_{n_1+n_2+n_3 = n,n_1 \neq 0}\frac{n_1}{n_1+ n_2} \binom{n}{n_3} \binom{n_1+n_2}{n_1}q_1^{n_1}q_2^{n_2}q_3^{n_3} \nonumber\\
	& = q_1\sum_{n_1+n_2+n_3 = n,n_1 \neq 0} \binom{n}{n_3} \binom{n_1+n_2-1}{n_1-1}q_1^{n_1-1}q_2^{n_2}q_3^{n_3} \nonumber\\
	& = q_1 \sum_{1\leq n_3 \leq n-1} \binom{n}{n_3}q_3^{n_3} [\sum_{n_1+n_2=n-n_3,n_1\neq 0}\binom{n_1+n_2-1}{n_1-1}q_1^{n_1-1}q_2^{n_2}]  \nonumber \\
	& = q_1 \sum_{1\leq n_3 \leq n-1} \binom{n}{n_3}q_3^{n_3} (q_1+q_2)^{n_1+n_2-1}\nonumber\\
	& = \frac{q_1}{q_1+ q_2}[(q_1+q_2+q_3)^n- q_3^n]\nonumber \\
	& = \frac{q_1}{q_1+ q_2}(1- q_3^n)\nonumber
		\end{align}

\textbf{Proof of Theorem 3.9}: The equality $\mathbb{E}[\frac{N_2 q- N_1p}{(N_1+N_2)(q-p)}] = \pi (1-q_3^n)$ follows directly from the above lemma and the formula for MLE.

\textbf{The proof of Theorem 3.10} is much more involved. It needs the following five necessary lemmas which are not stated in the main text. 
Recall that $n$ is the sample size.  Let  $A= \sum_{0\leq N_3 <n}\frac{1}{n-N_3}$ ${n}\choose{N_3}$$ (1-q_3)^{n-N_3}q_3^{N_3}$ and  $f(q_1, n_3)=\sum_{n_1+n_2=n-n_3,n_1\neq 0}n_1\binom{n_1+n_2-1}{n_1-1}q_1^{n_1-1}q_2^{n_2}$.

\begin{lemma} 
\begin{equation*}
f(q_1,n_3) = \left\{
\begin{array}{rl}
1 & \text{ if } n_3= n-1,\\
(q_1+q_2)^{n-n_3-1}+(n-n_3-1)q_1(q_1+q_2)^{n-n_3-2} & \text{ if } n_3 \leq n-2.
\end{array} \right.
\end{equation*}

 \end{lemma}
\begin{proof} First we compute the integral of $f(x, n_3)$.
	\begin{align}
	\int_{0}^{x} f(t, n_3) dt &= \sum_{n_1+n_2=n-n_3,n_1\neq 0}\int_{0}^{x}n_1\binom{n_1+n_2-1}{n_1-1}t^{n_1-1}q_2^{n_2}dt\nonumber \\
	& = \sum_{n_1+n_2=n-n_3,n_1\neq 0}\binom{n_1+n_2-1}{n_1-1}x^{n_1}q_2^{n_2}\nonumber\\
	& = x(x+q_2)^{n-n_3-1}\nonumber
	\end{align}

So $f(x,n_3) = (\int_0^x f(t)dt)'= (x+q_2)^{n-n_3-1} + (n-n_3-1)x(x+q_2)^{n-n_3-2}$. 
\end{proof}

\begin{lemma} $\sum_{0\leq n_3 \leq n-1} \frac{1}{n-n_3}(q_1+q_2)^{n-n_3-1}\binom{n}{n_3}q_3^{n_3} = \frac{A}{(q_1+q_2)}$.
\end{lemma}

\begin{lemma} $\sum_{0\leq n_3 \leq n-2} (q_1+q_2)^{n-n_3-2}\binom{n}{n_3}q_3^{n_3} = \frac{1-q_3^n-n(q_1+q_2)q_3^{n-1}}{(q_1+q_2)^2}$.
\end{lemma}
\begin{proof} It follows from the fact that 
\begin{align}
[(q_1+q_2)+q_3]^n  &= [\sum_{0\leq n_3 \leq n-2} (q_1+q_2)^{n-n_3-2}\binom{n}{n_3}q_3^{n_3}] + q_3^n+n(q_1+q_2)q_3^{n-1}.\nonumber 
\end{align}
\end{proof}

\begin{lemma} $\sum_{0\leq n_3 \leq n-2} \frac{1}{n-n_3} (q_1+q_2)^{n-n_3-2}\binom{n}{n_3}q_3^{n_3} = \frac{A-n(q_1+q_2)q_3^{n-1}}{(q_1+q_2)^2}$.
\end{lemma}
\begin{proof} It follows directly from the definition of $A$.
\end{proof}

 \begin{lemma}
\begin{itemize}
	\item $\mathbb{E}[(\frac{N_1}{N_1+N_2})^2] = \frac{q_1q_2}{(q_1+q_2)^2} A + \frac{q_1^2}{(q_1+q_2)^2}(1-q_3^n)$
	
	\item $\mathbb{E}[(\frac{N_2}{N_1+N_2})^2] = \frac{q_1q_2}{(q_1+q_2)^2} A + \frac{q_2^2}{(q_1+q_2)^2}(1-q_3^n)$.	
\end{itemize}

\end{lemma}

\begin{proof} We only show the first part and the proof of the second is similar. 
	\begin{align}
	\mathbb{E}[(\frac{N_1}{N_1+N_2})^2]  = & \sum_{n_1+n_2+n_3=n, n_1\neq 0} (\frac{n_1}{n_1+n_2})^2 \binom{n_1+n_2}{n_1}\binom{n}{n_3} q_1^{n_1}q_2^{n_2}q_3^{n_3} \nonumber\\
	=  & q_1 \sum_{n_1+n_2+n_3=n, n_1\neq 0}\frac{1}{n-n_3}n_1\binom{n_1+n_2-1}{n_1-1}\binom{n}{n_3}q_1^{n_1-1}q_2^{n_2}q_3^{n_3}\nonumber\\
	 = & q_1 \sum_{0\leq n_3 \leq n-1}\frac{1}{n-n_3}[\sum_{n_1+n_2=n-n_3,n_1\neq 0}n_1\binom{n_1+n_2-1}{n_1-1}q_1^{n_1-1}q_2^{n_2}]\binom{n}{n_3} q_3^{n_3}\nonumber\\
	= & q_1 [\sum_{0\leq n_3 \leq n-1} \frac{1}{n-n_3}(q_1+q_2)^{n-n_3-1}\binom{n}{n_3}q_3^{n_3} + q_1\sum_{0\leq n_3 \leq n-2} (q_1+q_2)^{n-n_3-2}\binom{n}{n_3}q_3^{n_3}-\nonumber \\
	&q_1\sum_{0\leq n_3 \leq n-2} \frac{1}{n-n_3} (q_1+q_2)^{n-n_3-2}\binom{n}{n_3}q_3^{n_3} ] \text{ (Lemma 0.2) }\nonumber\\
	= & q_1[\frac{A}{(q_1+q_2)} + q_1\frac{1-q_3^n-n(q_1+q_2)q_3^{n-1}}{(q_1+q_2)^2} -  q_1\frac{A-n(q_1+q_2)q_3^{n-1}}{(q_1+q_2)^2} ] \text{ (Lemmas 0.3, 0.4,0.5) }\nonumber \\
	=& \frac{q_1q_2}{(q_1+q_2)^2} A + \frac{q_1^2}{(q_1+q_2)^2}(1-q_3^n) \nonumber 
	\end{align}
\end{proof}

\textbf{Proof of Theorem 3}:
    \begin{align*}
        Var[\hat{\pi} | N_1+ N_2 \neq 0] &= Var[ \frac{N_2 q- N_1p}{(N_1+N_2)(q-p)}| N_1+ N_2 \neq 0]\\
        &=Var[ \frac{N_2 q+ N_2p-(N_1+N_2)p}{(N_1+N_2)(q-p)}|N_1+ N_2 \neq 0]\\
        &=Var[\frac{N_2}{N_1+N_2}\frac{p+q}{q-p}-\frac{p}{q-p} | N_1+ N_2 \neq 0]\\
        &=\frac{(p+q)^2}{(q-p)^2}Var[\frac{N_2}{N_1+N_2} | N_1+ N_2 \neq 0]\\
        &=\frac{(p+q)^2}{(q-p)^2}[E[(\frac{N_2}{N_1+N_2})^2| N_1+ N_2 \neq 0]-E[\frac{N_2}{N_1+N_2}|N_1+ N_2 \neq 0]^2]\\
        &= \frac{(p+q)^2}{(q-p)^2}[\frac{q_1q_2}{(q_1+q_2)^2} A + \frac{q_2^2}{(q_1+q_2)^2}-\frac{q_2^2}{(q_1+q_2)^2}] \text{ (Lemma 0.6) }\\
        &=\frac{1}{(q-p)^2}q_1q_2A\\
        &=\frac{1}{(q-p)^2}[\pi p+(1-\pi)q][\pi q+(1-\pi)p]A
    \end{align*}
%\end{proof}

\textbf{Proof of Corollary 3.11}: $f'(q) = \frac{(n+1)(q+3p)+(\pi-1/2)(p+q)}{(p-q)^3(n+1)} >0$. 

\textbf{Proof of Corollary 3.12}: When $p+q$ is a constant, the optimaility of $Q_{GWRR}$ follows from the same arguments for $G_{WRR}$ (Example 1) in \cite{Warner65} and \cite{HolohanLM17}.

\textbf{Proof of Lemma 3.14}:  $\mathcal{P}_{bel_x} = \{pr: pr\text{ is a probability distribution and }pr(A)\geq bel_x(A)\text{ for all } A\}$ and $\mathcal{P}_{bel_{x'}} = \{pr: pr\text{ is a probability distribution and }pr(A)\geq bel_{x'}(A)\text{ for all } A\}$. For a given $E\subseteq Y$, by using specialization matrices in \cite{KlawonnS92}, we can always find a $pr_x\in \mathcal{P}_{bel_x}$ such that $pr_x(E) = pl_x(E)$ and a $pr_{x'}$ such that $pr_{x'} (E) = bel_{x'}(E)$. And Lemma 5 follows directly. 

\textbf{Proof of Lemma 3.15}: The proof is similar to that for Lemma 2.

\textbf{Proof of Lemma 3.16}: According to the well-known weak von-Neumann-Birkhoff Lemma, we only need to prove the proposition only for a deterministic function $f: Y\rightarrow Z$. The proof for $\epsilon$-SLDP algorithm is similar to that of the post-processing property for probabilistic privacy mechanisms. For $\epsilon$-WLDP, consider the two probability functions $p_x$ and $p_{x'}$ such that 
\begin{align}
    r^Q_W = \frac{p_x(E)}{p_{x'}(E)}
\end{align}
for some $E$.  Then 
\begin{align*}
    r^{f\circ Q}_W &= max_{x,x'\in X, E\subseteq Z} \frac{p_xf^{-1}(E)}{p_{x'}f^{-1}(E)}\\
    & \leq max_{x,x'\in X, E\subseteq Z} \frac{p_x(E)}{p_{x'}(E)}\\
    & \leq r^Q_W.
\end{align*}

%\textbf{Proof of Lemma 8}:  It follows from the following observation that, for any $x_3\in X_3$, $pr_x^{Q_2\circ Q_1} = pr_x^{Q_1} (Q_2^{-1}(x_3)) pr_{Q_2^{-1}(x_3)}^{Q_2}(x_3) \leq e^{\epsilon_1+\epsilon_2}$.

%&&&&&&&&&&&&&&&&&&&&&&&&&&&&&&&&&&&&&&&&&&&&&&&&&&&&&&&&&&&&&&&&&&&&&&&&&&&&&&&&&&7

\textbf{Proof of Theorem 3.18}  Let $R$ be a rejection region. Since $Q$ is $\epsilon$-WLDP, the following inequalities hold:
\begin{align}
	e^{-\epsilon} & \leq \frac{pl_{x'}^Q(R)}{bel_x^Q(R)}, \frac{pl_{x}^Q(R)}{bel_{x'}^Q(R)} \leq  e^{\epsilon}\\
	e^{-\epsilon} & \leq \frac{pl_{x'}^Q(R^c)}{bel_x^Q(R^c)}, \frac{pl_{x}^Q(R^c)}{bel_{x'}^Q(R^c)}  \leq  e^{\epsilon}
\end{align}
\begin{enumerate}
	\item Assume that $pl_x^Q(R)\leq \alpha$.  It follows that $e^{-\epsilon}bel_{x'}^Q(R)\leq pl_x^Q(R)\leq \alpha$.  This implies that $bel_{x'}^Q(R) \leq e^{\epsilon}\alpha$. So $pl_{x'}^Q(R^c)\geq 1-e^{\epsilon}\alpha$. Moreover, Since $bel_x^Q(R^c)\leq e^{\epsilon}pl_{x'}^Q(R^c)$, $1-e^{\epsilon}pl_{x'}^Q(R^c)\leq 1-bel_x^Q(R^c) = pl_x^Q(R)\leq \alpha$. So we have that $pl_{x'}(R^c)\geq e^{-\epsilon}(1-x)$. Putting all these togther, we obtain that $T^{pe}(Q(x), Q(x')) (\alpha)\geq max\{1-\alpha e^{\epsilon}, 0, e^{-\epsilon}(1-\alpha)\}$. 

	\item Assume that $pl_x^Q(R^c)\geq \alpha$. Since $pl_x^Q(R^c)\leq e^{\epsilon}bel_{x'}^Q(R^c)= e^{\epsilon}(1-pl_{x'}^Q(R))$, $pl_{x'}^Q(1-\alpha)$.  Moreover, $\alpha \leq pl_{x'}^Q(R^c)\leq e^{\epsilon} bel_x^Q(R^c) = e^{\epsilon}(1-pl_x^Q(R))$.  It follows that $pl_{x'}^Q(R)\leq 1-e^{-\epsilon}\alpha$. So $T^{op}(Q(x), Q(x'))(\alpha)\leq  min\{1-\alpha e^{-\epsilon}, e^{\epsilon}(1-\alpha)\}$.

\end{enumerate}
The other direction of the equivalence can shown in a similar way. We have finished the proof of Theorem 4.

\textbf{Proof of Corollary 3.19} The proof is similar to that for Corollary 1.  The only things that we need to pay much attention to is the direction of the inequalities.

%&&&&&&&&&&&&&&&&&&&&&&&&&&&&&&&&&&&&&&&&&&&

\textbf{Proof of Corollary 3.20}: From the above analysis, we know that $\epsilon^W(p,q)$ is defined as the privacy loss of the following matrix $Q_1$ in Warner's model:

\begin{equation*}
Q_{1}= \left(
\begin{array}{ccc}
1-q & q  \\
q  & 1-q 
\end{array} \right)
\end{equation*}
So $\epsilon^W(p,q)= ln(\frac{1-q}{q})$ and it is decreasing with respct to $q$.  And $\nu^W(p,q)$ is defined  as the accuracy of the following matrix $Q_0$ in Warner's model:
\begin{equation*}
Q_{0}= \left(
\begin{array}{ccc}
p & 1-p  \\
1-p  & p 
\end{array} \right)
\end{equation*}
So $\nu^W(p,q)= \frac{-(\pi-1/2)^2 +\frac{1}{4}(\frac{1}{2p-1})^2}{n}$ and it is decreasing with respect to $p$.

    \end{document}